\pdfoutput=1

\documentclass[final]{siamltex}

\usepackage{graphicx,times,amsmath}
\usepackage{subfigure} 

\usepackage{amsmath}
\usepackage{amsfonts}
\usepackage{amssymb}
\usepackage{stmaryrd}

\newcommand{\R}{\mathbb R}
\newcommand{\mmu}{\boldsymbol{\mu}}
\newcommand{\llambda}{\boldsymbol{\lambda}}

\title{Min Max Generalization for Two-stage Deterministic Batch Mode Reinforcement Learning: Relaxation Schemes}

\author{R. Fonteneau\footnotemark[2] 
\and D. Ernst\footnotemark[2] 
\and B. Boigelot\footnotemark[2]
\and Q. Louveaux\footnotemark[2]
}

\begin{document}

\maketitle

\renewcommand{\thefootnote}{\fnsymbol{footnote}}
\footnotetext[2]{Department of Electrical Engineering and Computer Science, University of Li\`{e}ge, Belgium}
\renewcommand{\thefootnote}{\arabic{footnote}}

\begin{abstract}
We study the $\min \max$ optimization problem introduced in \cite{Fonteneau2011minmax} for computing policies for batch mode reinforcement learning in a deterministic setting. First, we show that this problem is NP-hard. In the two-stage case, we provide two relaxation schemes. The first relaxation scheme works by dropping some constraints in order to obtain a problem that is solvable in polynomial time. The second relaxation scheme, based on a Lagrangian relaxation where all constraints are dualized, leads to a conic quadratic programming problem. We also theoretically prove and empirically illustrate that both relaxation schemes provide better results than those given in \cite{Fonteneau2011minmax}.
\end{abstract}



\pagestyle{myheadings}
\thispagestyle{plain}
\markboth{R. Fonteneau, D. Ernst, B. Boigelot, Q. Louveaux}{Min Max Generalization for Two-stage Deterministic Batch Mode RL: Relaxation Schemes}

\section{Introduction}
\label{section:introduction}

Research in Reinforcement Learning (RL) \cite{Sutton1998} aims at designing computational agents able to learn by themselves how to interact  with their environment to maximize a numerical reward signal. The techniques developed in this field have  appealed researchers  trying to solve  sequential decision making problems  in many fields such  as Finance \cite{Ingersoll1987}, Medicine \cite{Murphy2003,Murphy2005} or Engineering \cite{Riedmiller2005}. Since the end of the nineties, several researchers have focused on the resolution of a subproblem of RL: computing a high-performance policy when the only information available on the environment is contained in a batch collection of trajectories of the agent \cite{Bradtke1996,Ernst2005,Lagoudakis2003,Ormoneit2002,Riedmiller2005,Fonteneau2011Thesis}. This subfield of RL is known as ``batch mode RL".

Batch mode RL (BMRL) algorithms are challenged when dealing with large or continuous state  spaces. Indeed, in such cases they have to generalize the information contained in a generally sparse sample of trajectories. The dominant approach for generalizing this information is to combine BMRL algorithms with function approximators \cite{Bertsekas1996,Lagoudakis2003,Ernst2005,Busoniu2010}. Usually, these approximators generalize the information contained in the sample to areas poorly  covered by the sample by implicitly assuming that the properties of the system in those areas are similar to the properties of the system in the nearby areas well covered by the sample. This in turn often leads to low performance guarantees on the inferred policy when large state space areas are poorly covered by the sample. This can be explained by the fact that when computing the performance guarantees of these policies, one needs to take into account that they may actually drive the system into the poorly visited areas to which the generalization strategy associates a favorable  environment behavior, while the environment may  actually be particularly adversarial in those areas. This is corroborated by theoretical results which show that the performance guarantees of  the policies inferred by  these algorithms degrade with the  sample dispersion where, loosely  speaking, the  dispersion  can be seen as the radius of the largest non-visited state space area.

To overcome this problem, \cite{Fonteneau2011minmax}  propose  a $\min\max$-type strategy  for generalizing in deterministic, Lipschitz continuous  environments with continuous state  spaces,  finite action spaces, and  finite time-horizon.   The $\min\max$ approach  works by determining  a sequence of  actions that maximizes the worst return that could possibly be obtained considering any  system  compatible  with  the sample  of trajectories, and  a weak prior knowledge  given in the  form of upper bounds on the Lipschitz constants related to the environment (dynamics, reward function). However, they show that finding an  exact solution of the $\min\max$  problem is far from trivial,  even after  reformulating the  problem  so as  to avoid  the search in the space of  all compatible functions.  To circumvent these difficulties, they  propose to replace, inside  this $\min\max$ problem, the search for the worst environment given a sequence of actions by an expression that lower-bounds the worst possible return which leads to their so called CGRL algorithm (the acronym stands for ``Cautious approach to Generalization in Reinforcement Learning''). This  lower  bound is  derived  from their previous work \cite{Fonteneau2009ADPRL,Fonteneau2010ICAART}  and has  a  tightness that  depends on  the sample  dispersion. However, in some configurations where areas of the the state space are not well covered by the sample of trajectories, the CGRL bound turns to be very conservative.

In this paper, we propose to further investigate the $\min \max$ generalization optimization problem that was initially proposed in \cite{Fonteneau2011minmax}. We first show that this optimization problem is NP-hard. We then focus on the two-stage case, which is still NP-hard. Since it seems hopeless to exactly solve the problem, we propose two relaxation schemes that preserve the nature of the $\min \max$ generalization problem by targetting policies leading to high performance guarantees. The first relaxation scheme works by dropping some constraints in order to obtain a problem that is solvable in polynomial time. This results into a well known configuration called the {\it trust-region subproblem} \cite{Conn2000}. The second relaxation scheme, based on a Lagrangian relaxation where all constraints are dualized, can be solved using conic quadratic programming  in polynomial time. We prove that both relaxation schemes always provide bounds that are greater or equal to the CGRL bound. We also show that these bounds are tight in a sense that they converge towards the actual return when the sample dispersion converges towards zero, and that the sequences of actions that maximize these bounds converge towards optimal ones.

The paper is organized as follows:
\begin{itemize}
\item in Section \ref{section:related_work}, we give a short summary of the literature  related to this work,
\item Section \ref{section:problem_formalization} formalizes the $\min \max$ generalization problem in a Lipschitz continuous, deterministic BMRL context,
\item in Section  \ref{section:2_stage_case}, we  focus on the particular two-stage  case, for which we prove that it can be decoupled into two independent problems corresponding respectively to the first stage and the second stage (Theorem \ref{theorem_decoupling}):
\begin{itemize}
\item  the first stage problem leads to a trivial optimization problem that can be solved in closed-form (Corollary \ref{corollary:solution_P'2}),
\item we prove in Section  \ref{subsection:complexity} that the second stage problem is NP-hard (Corollary \ref{corollary:complexity_second_stage}), which consequently proves the NP-hardness of the general $\min \max$ generalization problem (Theorem \ref{theorem:complexity}),
\end{itemize}
\item we then describe in Section \ref{section:relaxation_schemes_2_stage} the two relaxation schemes that we propose for the second stage problem:  
\begin{itemize}
\item the trust-region relaxation scheme (Section \ref{subsection:trust_region}),
\item the Lagrangian relaxation scheme  (Section \ref{subsection:lagrangian_relaxation}), which is shown to be a conic-quadratic problem (Theorem \ref{theorem:conic_quadratic}),
\end{itemize}
\item we prove in Section \ref{subsubsection:TRversusCGRL}  that the first relaxation scheme gives better results than CGRL (Theorem \ref{theorem:TRversusCGRL}),
\item we show in Section \ref{subsubsection:LDversusTR} that the second relaxation scheme povides better results than the first relaxation scheme (Theorem \ref{theorem:LDversusTR}), and consequently better results than CGRL (Theorem \ref{theorem:synthesis}),
\item we  analyze in Section \ref{subsection:convergence_bounds} the asymptotic behavior of the  relaxation schemes as a function of the sample dispersion:
\begin{itemize}
\item we show that the the bounds provided by the relaxtion schemes converge towards the actual return when the sample dispersion decreases towards zero (Theorem \ref{theorem:consistency_bounds}),
\item we show that the sequences of actions maximizing such bounds converge towards  optimal sequences of actions when the sample dispersion decreases towards zero (Theorem \ref{theorem_convergence_actions}),
\end{itemize}
\item Section \ref{section:experimental_results} illustrates the relaxation schemes on an academic benchmark, 
\item Section \ref{section:conclusions} concludes.
\end{itemize}
We provide in Figure \ref{fig:roadmap} an illustration of the roadmap of the main  results of this paper. 
\begin{figure}
\begin{center}
\includegraphics[width=.6\linewidth]{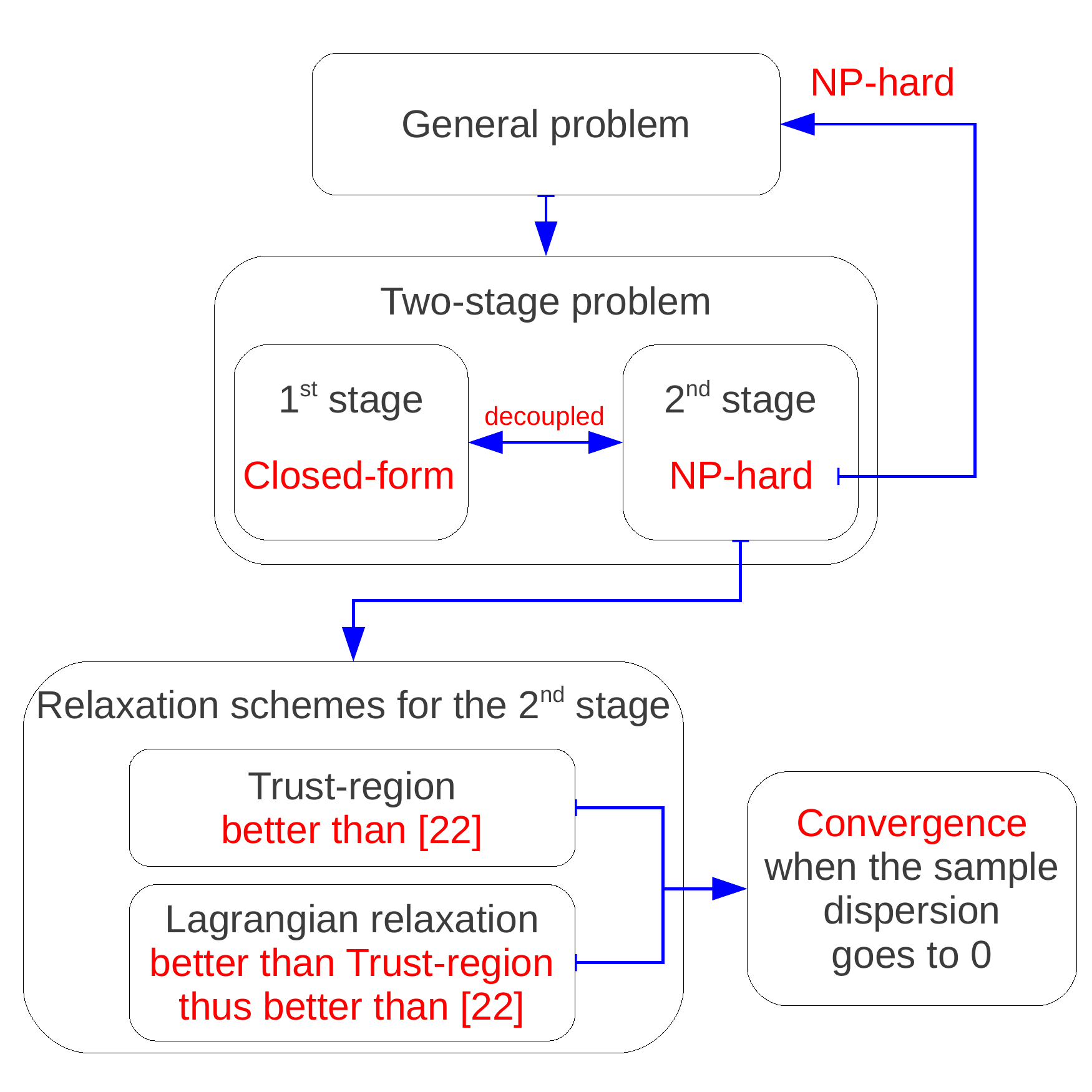}
\caption{Main results of the paper.}
\label{fig:roadmap}
\end{center}
\end{figure}

\section{Related Work}
\label{section:related_work}

Several works have already been built upon $\min\max$ paradigms for computing policies in a RL setting. In stochastic frameworks, $\min \max$ approaches are often successful for deriving robust solutions with respect to uncertainties in the (parametric) representation of the probability distributions associated with the environment \cite{Delage2010}. In the context where several agents interact with each other in the same environment, $\min \max$ approaches appear to be efficient strategies for designing policies that maximize one agent's reward given the worst adversarial behavior of the other agents. \cite{Littman1994,Rovatsou2010}. They have also received some attention for solving partially observable Markov decision processes \cite{Littman2009POMDP,Koenig2001}.

The $\min \max$ approach towards generalization, originally introduced in \cite{Fonteneau2011minmax}, implicitly relies on a methodology for computing lower bounds  on  the worst possible  return   (considering  any  compatible environment)  in  a deterministic setting  with a  mostly unknown actual  environment. In this respect, it is related to other approaches that aim at computing performance guarantees on the returns of inferred policies \cite{Mannor2004,Qian2009,Paduraru2011}.

Other fields of research have proposed $\min \max$-type strategies for computing control policies. This includes Robust Control theory \cite{Hansen2001} with H$_{\infty}$ methods \cite{Basar1995}, but also Model   Predictive Control  (MPC) theory - where usually the environment  is supposed  to be  fully known \cite{Camacho2004,Ernst2009SMC} - for which $\min\max$ approaches have been used to determine an optimal  sequence of  actions with respect  to the  ``worst case'' disturbance sequence occurring \cite{Scokaert1998,Bemporad1999}. Finally, there is a broad stream of works in the field of Stochastic Programming \cite{BirgeLouveaux1997} that have addressed the problem of safely planning under uncertainties, mainly known as ``robust stochastic programming" or ``risk-averse stochastic programming'' \cite{Defourny2008,Shapiro2011dynamic,Shapiro2011minimax,Nemirovski2009Robust}. In this field, the two-stage case has also been particularly well-studied \cite{Frauendorfer1992,Darby2000}.

\section{Problem Formalization}
\label{section:problem_formalization}

We first formalize the BMRL setting in Section \ref{subsection:batch_mode_rl}, and we state the $\min \max$ generalization problem in Section \ref{subsection:min_max_generalization}.

\subsection{Batch Mode Reinforcement Learning}
\label{subsection:batch_mode_rl}

We consider a deterministic discrete-time system whose dynamics over $T$ stages is described by  a time-invariant  equation 
\begin{eqnarray*}
x_{t+1} =  f\left(x_t,u_t\right) \quad t=0,\ldots,T-1,
\end{eqnarray*}
 where for  all $t$, the state $x_t$  is an element of the  state space $\mathcal X  \subset \mathbb R^{d}$  where  $\mathbb R^{d}$  denotes the  $d-$dimensional Euclidean space  and  $u_t$ is  an  element of  the finite (discrete)  action space $\mathcal U = \left\{ u^{(1)}, \ldots, u^{(m)}  \right\}$ that we abusively identify with $\left\{1, \ldots, m \right\}$.   $T\in{\mathbb N} \setminus \{0\}$ is referred  to as  the (finite) optimization horizon.   An instantaneous  reward
\begin{eqnarray*}
r_t=\rho\left(x_t,u_t\right)\in  {\mathbb  R}
\end{eqnarray*}
is  associated with  the  action $u_t$ taken while  being in state $x_{t}$.  For  a given initial state $x_0   \in   \mathcal   X$   and   for   every   sequence   of   actions $\left(u_0,\ldots, u_{T-1}\right) \in  \mathcal U^T  $, the cumulated  reward over $T$   stages  (also  named  $T-$stage return)   is  defined as follows:
\begin{definition}[$T-$stage Return]
\begin{eqnarray}
\forall \left( u_0,\ldots, u_{T-1}\right) \in \mathcal U^T, \qquad J^{(u_0,\ldots,u_{T-1})}_T \triangleq  \sum_{t=0}^{T-1} \rho\left(x_t, u_t \right) \ , \nonumber
\end{eqnarray}
where
\begin{eqnarray}
x_{t+1} = f\left(x_t, u_t \right) \ , \qquad \forall t \in \{0, \ldots, T-1 \} \ .  \nonumber
\end{eqnarray}
\end{definition}
An optimal sequence of actions is a sequence that leads to the maximization of the $T-$stage return:
\begin{definition}[Optimal $T-$stage Return]
\begin{eqnarray}
J^{*}_T \triangleq \underset {(u_0,\ldots, u_{T-1}) \in \mathcal U^T} {\max} J^{(u_0,\ldots, u_{T-1})}_T \ .  \nonumber
\end{eqnarray}
\end{definition}
We further make the following assumptions that characterize the {\it batch mode setting}:
\begin{enumerate}
\item The system  dynamics $f$ and the  reward function  $\rho$ are {\it unknown};
\item For each action $u \in \mathcal U$, a set  of $n^{(u)} \in \mathbb N$ one-step system transitions $$\mathcal F^{(u)} = \left\{\left(x^{(u),k},r^{(u),k},y^{(u),k}\right) \right\}_{k=1}^{n^{(u)}}$$ is known where  each  one-step transition  is  such  that:
\begin{eqnarray}
y^{(u),k}=f\left(x^{(u),k}, u \right)\mbox{  and }r^{(u),k}=\rho\left(x^{(u),k}, u \right) \ . \nonumber
\end{eqnarray}
\item We assume that every set $\mathcal{F}^{(u)}$ contains at least one element:
\begin{eqnarray}
\forall u \in \mathcal U, \qquad n^{(u)} > 0 \ . \nonumber
\end{eqnarray}
\end{enumerate}
In the following, we denote by $\mathcal F$ the collection of all system transitions:
\begin{eqnarray}
\mathcal F = \mathcal F^{(1)} \cup \ldots \cup \mathcal F^{(m)}  \nonumber
\end{eqnarray}
Under those assumptions, batch mode reinforcement learning (BMRL) techniques propose to infer from the sample of one-step system transitions $\mathcal F$ a high-performance sequence of actions, i.e. a sequence of actions 
$\left( \tilde u^*_0 , \ldots,  \tilde u^*_{T-1} \right) \in \mathcal U^T$
such that $J^{(\tilde u^*_0, \ldots, \tilde u^*_{T-1})}_T$ is as close as possible to $J^*_T$.

\subsection{Min max Generalization under Lipschitz Continuity Assumptions}
\label{subsection:min_max_generalization}

In this section, we state the $\min \max$ generalization problem that we study in this paper. The formalization was originally proposed in \cite{Fonteneau2011minmax}.

We first assume that the system dynamics $f$ and the reward function $\rho$ are assumed to be Lipschitz continuous. There  exist finite constants  $L_f,L_\rho\in {\mathbb R}$ such  that:
\begin{eqnarray}
\forall (x, x') \in \mathcal X^2, \forall u \in \mathcal U, \qquad \left\|f\left(x, u \right) - f\left(x', u \right) \right\| & \leq & L_f  \left\| x - x' \right\| \
  ,\label{eqn:f_lipschitz} \nonumber \\
  \left| \rho\left(x, u \right) - \rho\left(x', u \right) \right| & \leq & L_\rho  \left\| x - x' \right\| \
  , \label{eqn:rho_lipschitz} \nonumber 
\end{eqnarray}
where $\| . \|$ denotes the Euclidean norm over the space $\mathcal X$. We also assume that two constants $L_f$  and  $L_\rho$  satisfying  the above-written  inequalities  are known.

For a given sequence of actions, one can define the worst possible return that can be obtained by any system whose dynamics $f'$ and $\rho'$ would satisfy the Lipschitz inequalities and that would coincide with the values of the functions $f$ and $\rho$ given by the sample of system transitions $\mathcal F$. As shown in \cite{Fonteneau2011minmax}, this worst possible return can be computed by solving a  finite-dimensional optimization  problem  over  $\mathcal  X^{T-1}  \times  \mathbb R^{T}$. Intuitively, solving such an optimization problem amounts in determining a most pessimistic trajectory of the system that is still compliant with the sample of data and the Lipschitz continuity assumptions.
More specifically, for a given sequence of actions $(u_0,\ldots, u_{T-1}) \in {\mathcal U}^T$, some given constants $L_f$ and $L_\rho$, a given initial state $x_0 \in \mathcal X$ and a given sample of transitions ${\mathcal F}$,  this optimization problem writes:
\\
\\
\fbox{
\begin{centering}
\begin{minipage}{0.95\textwidth}
$ \left( \mathcal P_T(\mathcal F,L_f,L_\rho,x_0,u_0,\ldots,u_{T-1}) \right):$
\begin{eqnarray}
\underset {\begin{matrix} \mathbf{ \hat r_0} & \ldots & \mathbf{\hat r_{T-1}} \in \mathbb R\\
\mathbf{\hat     x_0} &\ldots&  \mathbf{\hat    x_{T-1}} \in \mathcal X\end{matrix}} {\min} \qquad
  \sum_{t=0}^{T-1} \mathbf{\hat r_t} , \nonumber 
\end{eqnarray}
subject to
\begin{align*}
&\left| \mathbf{\hat r_t} - r^{(u_t),k_t} \right|^2 \leq L_\rho^2 \left\| \mathbf{\hat x_t} - x^{(u_t),k_t}  \right\|^2 \ ,  \forall (t,k_t) \in  \left\{0 , \ldots, T-1 \right\} \times \left\{1,\ldots,n^{(u_t)} \right\}, \\  
&\left\| \mathbf{\hat x_{t+1}} - y^{(u_t),k_t} \right\|^2 \leq L_f^2 \left\| \mathbf{\hat x_{t}} - x^{(u_t),k_t} \right\|^2 \hspace{-.2cm} , \forall (t,k_t)  \in \left\{0 , \ldots, T-1 \right\} \times \left\{1,\ldots,n^{(u_t)} \right\}\hspace{-.1cm}, \\ 
&  \left| \mathbf{\hat r_{t}} - \mathbf{\hat r_{t'}}  \right|^2  \leq L_\rho^2 \left\|\mathbf{\hat x_t}  -  \mathbf{\hat x_{t'}}  \right\|^2  \ ,  \forall t,t' \in \left\{0,\ldots, T-1  | u_t = u_{t'} \right\}, \\
& \left\| \mathbf{\hat x_{t+1}} - \mathbf{\hat x_{t'+1}}   \right\|^2  \leq L_f^2 \left\|\mathbf{\hat x_t}  -  \mathbf{\hat x_{t'}}  \right\|^2 \ ,   \forall t,t' \in \left\{0,\ldots, T-2  | u_t = u_{t'} \right\},  \\ 
& \left. \mathbf{\hat x_0} = x_0 \ . \right.    
\end{align*}
\end{minipage}
\end{centering}
}
\\
\\
Note that, throughout the paper,  optimization variables will be written in bold.

The min max approach to generalization aims at identifying which sequence of actions maximizes its worst possible return, that is which sequence of actions leads to the highest value of $\left( \mathcal P_T(\mathcal F,L_f,L_\rho,x_0,u_0,\ldots,u_{T-1}) \right)$.

We focus in this paper on the design of resolution schemes for solving the program  $\left( \mathcal P_T(\mathcal F,L_f,L_\rho,x_0,u_0,\ldots,u_{T-1}) \right)$. 
These schemes can afterwards be used for solving the $\min \max$ problem through exhaustive search over  the set of all sequences of actions.

Later in this paper, we will also analyze the computational complexity of  this min max generalization problem. When carrying out this analysis, we will assume that all the data of the problem (i.e., $T,\mathcal F,L_f,L_\rho,x_0,u_0,\ldots,u_{T-1}$) are given in the form of rational numbers.

\section{The Two-stage Case}
\label{section:2_stage_case}

In this section, we restrict ourselves to the case where the time horizon contains only two steps, i.e. $T=2$, which is an important particular case of $\big(\mathcal P_T(\mathcal F,L_f,L_\rho,x_0,u_0,\ldots,u_{T-1})\big)$. Many works in optimal sequential decision making have considered the two-stage case \cite{Frauendorfer1992,Darby2000}, which relates to many applications, such as for instance medical applications where one wants to infer ``safe'' clinical decision rules from batch collections of clinical data \cite{Banerjee2006,Lokhnygina2008,Lunceford2002,Wahed2004}.

In Section \ref{subsection:decoupling_P_2}, we show that this problem can be decoupled into two subproblems. While the first subproblem is straightforward to solve, we prove in Section \ref{subsection:complexity} that the second one is NP-hard, which proves that the two-stage problem as well as the generalized $T-$stage problem $\left(\mathcal P_T(\mathcal F,L_f,L_\rho,x_0,u_0,\ldots,u_{T-1})\right)$ are also NP-hard.

Given a two-stage sequence of actions $\left( u_0, u_1 \right) \in \mathcal U^2$, the two-stage version of the problem $\big(\mathcal P_T(\mathcal F,L_f,L_\rho,x_0,u_0,\ldots,u_{T-1})\big)$ writes as follows:\\
\\
 \fbox{
\begin{centering}
\begin{minipage}{0.95\textwidth}
$\big( \mathcal P_2(\mathcal F,L_f,L_\rho,x_0,u_0,u_1) \big):$
\begin{eqnarray}
\underset {\begin{matrix} \mathbf{\hat r_0} , \mathbf{\hat r_{1}} \in \mathbb R \\ \mathbf{\hat  x_0} , \mathbf{\hat   x_{1}} \in \mathcal X \end{matrix}} {\min} \qquad  \mathbf{\hat r_0} + \mathbf{\hat r_1}, \nonumber
\end{eqnarray}
subject to 
\begin{eqnarray}
&& \left|\mathbf{\hat r_0} - r^{(u_0),k_0} \right|^2 \leq L_\rho^2 \left\| \mathbf{\hat x_0} - x^{(u_0),k_0}  \right\|^2  \ ,  \forall  k_0 \in \left\{1, \ldots, n^{(u_0)} \right\} \ ,  \label{setConst_r_0_general} \\
&& \left| \mathbf{\hat r_1} - r^{(u_1),k_1} \right|^2  \leq L_\rho^2 \left\| \mathbf{\hat x_1} - x^{(u_1),k_1}  \right\|^2  \ , \forall  k_1 \in \left\{1, \ldots, n^{(u_1)} \right\} \ , \label{setConst_r_1_general}  \\
&& \left\| \mathbf{\hat x_{1}} - y^{(u_0),k_0} \right\|^2  \leq L_f^2 \left\| \mathbf{\hat x_{0}} - x^{(u_0),k_0} \right\|^2  \ , \forall  k_0 \in \left\{1, \ldots, n^{(u_0)} \right\} \ ,  \label{setConst_x_1_general} \\
&& \left| \mathbf{\hat r_0} - \mathbf{\hat r_1} \right|^2 \leq L_{\rho}^2 \left\| \mathbf{\hat x_0} - \mathbf{\hat x_1} \right\|^2  \mbox{ if } u_0 = u_1 \ , \label{setConst_r_0_r_1_general} \\
&& \mathbf{\hat x_0} = x_0 \ . \label{setConst_x_0_general}
\end{eqnarray}
\end{minipage}
\end{centering}
}\\

For a matter of simplicity, we will often drop the arguments in the definition of the optimization problem and refer   $\big(\mathcal P_2(\mathcal F,L_f,L_\rho,x_0,u_0,u_{1})\big)$ as $( \mathcal P^{(u_0,u_1)}_2 )$. We denote by $B^{(u_0,u_1)}_{2}(\mathcal F)$ the lower bound associated with an optimal solution of $( \mathcal P^{(u_0,u_1)}_2 ) $:
\begin{definition}[Optimal Value $B^{(u_0,u_1)}_{2}(\mathcal F)$]
Let $(u_0, u_1) \in \mathcal U^2$, and let $\left(\mathbf{\hat r_0^*},\mathbf{\hat r_1^*},\mathbf{\hat x_0^*},\mathbf{\hat x_1^*}\right)$ be an optimal solution to $\left( \mathcal P^{(u_0,u_1)}_2 \right). $ Then,
\begin{eqnarray*}
B^{(u_0,u_1)}_{2}(\mathcal F) \triangleq \mathbf{\hat r_0^*} + \mathbf{\hat r_1^*} \ .
\end{eqnarray*}
\end{definition}

\subsection{Decoupling Stages}
\label{subsection:decoupling_P_2}
Let $(\mathcal P'^{(u_0,u_1)}_2 )$ and $(\mathcal P''^{(u_0,u_1)}_2 )$ be the two following subproblems:
\\
\fbox{
\begin{centering}
\begin{minipage}{0.95\textwidth}
$ \left( \mathcal P'^{(u_0,u_1)}_2 \right) :$
\begin{eqnarray*}
\underset {\begin{matrix} \mathbf{\hat r_{0}} \in \mathbb R\\
 \mathbf{\hat    x_{0}} \in \mathcal X\end{matrix} } {\min}\qquad & \mathbf{\hat r_0} \nonumber
\end{eqnarray*}
subject to
\begin{eqnarray*}
&&\left | \mathbf{\hat r_0} - r^{(u_0),k_0} \right|^2 \leq L_\rho^2 \left\| \mathbf{\hat x_0} - x^{(u_0),k_0}  \right\|^2    \ , \forall  k_0 \in \left\{1, \ldots, n^{(u_0)} \right\} \ , \nonumber   \\
&&\mathbf{\hat x_0} = x_0 \ .   \nonumber
\end{eqnarray*}
\end{minipage}
\end{centering}
}\\
\\
\\
\fbox{
\begin{centering}
\begin{minipage}{0.95\textwidth}
$ \left( \mathcal P''^{(u_0,u_1)}_2 \right) :$
\begin{eqnarray}
 \underset {\begin{matrix} \mathbf{\hat r_{1}} \in \mathbb R\\
 \mathbf{\hat    x_{1}} \in \mathcal X\end{matrix}} {\min} \qquad  \mathbf{\hat r_1}
\end{eqnarray}
subject to
\begin{eqnarray}
&&\left| \mathbf{\hat r_1} - r^{(u_1), k_1} \right|^2 \leq L_\rho^2 \left\| \mathbf{\hat x_1} - x^{(u_1), k_1}  \right\|^2 \ \ ,  \forall k_1 \in \left\{1, \ldots, n^{(u_1)} \right\} \ , \label{c1p2''} \\
&&\left\| \mathbf{\hat x_{1}} - y^{(u_0), k_0} \right\|^2 \leq L_f^2 \left \|  x_{0} - x^{(u_0), k_0} \right\|^2 \ , \forall  k_0 \in \left\{1, \ldots, n^{(u_0)} \right\} \ . \label{c2p2''}
\end{eqnarray}
\end{minipage}
\end{centering}
}\\
\\

We show in this section that an optimal solution to $(\mathcal P^{(u_0,u_1)}_2 )$ can be obtained by solving the two subproblems $ ( \mathcal P'^{(u_0,u_1)}_2 )$ and $( \mathcal P''^{(u_0,u_1)}_2 )$ corresponding to the first stage and the second stage. Indeed,  one can see that the stages $t = 0$ and $t = 1$ are theoretically coupled by constraint  \eqref{setConst_r_0_r_1_general}, except in the case where the two actions $u_0$ and $u_1$ are different for which $(\mathcal P^{(u_0,u_1)}_2 )$ is trivially decoupled. We prove in the following that, even in the case $u_0 = u_1$, optimal solutions to the two decoupled problems $ ( \mathcal P'^{(u_0,u_1)}_2 )$ and $( \mathcal P''^{(u_0,u_1)}_2 )$ also satisfy constraint  \eqref{setConst_r_0_r_1_general}. Additionally, we provide the solution of $ ( \mathcal P'^{(u_0,u_1)}_2 )$. 

\begin{theorem}
\label{theorem_decoupling}
Let $\left( u_0, u_1 \right) \in \mathcal U^2$. If $\left(\mathbf{\hat r_0^*},\mathbf{\hat x_0^*} \right)$ is an optimal solution to $\left(\mathcal P'^{(u_0,u_1)}_2 \right)$ and $\left(\mathbf{\hat r_1^*},\mathbf{\hat x_1^*} \right)$ is an optimal solution to $\left( \mathcal P''^{(u_0,u_1)}_2 \right)$, then  $\left(\mathbf{\hat r_0^*},\mathbf{\hat r_1^*},\mathbf{\hat x_0^*},\mathbf{\hat x_1^*}\right)$ is an optimal solution to $\left(\mathcal P^{(u_0,u_1)}_2\right)$.
\end{theorem}

\begin{proof}
\begin{itemize}
\item First case:  $u_0 \neq u_1  \ .$
\end{itemize}
The constraint \eqref{setConst_r_0_r_1_general} drops and the theorem is trivial.

\begin{itemize}
\item Second case:  $u_0 = u_1 \ .$
\end{itemize}
The rationale of the proof is the following.  We first relax constraint \eqref{setConst_r_0_r_1_general}, and consider the two problems $( \mathcal P'^{(u_0,u_1)}_2)$ and $( \mathcal P''^{(u_0,u_1)}_2)$. Then, we show that optimal solutions of $( \mathcal P'^{(u_0,u_1)}_2)$ and $( \mathcal P''^{(u_0,u_1)}_2)$ also satisfy constraint \eqref{setConst_r_0_r_1_general}.

\paragraph{About $( \mathcal P'^{(u_0,u_1)}_2 )$}
The problem $( \mathcal P'^{(u_0,u_1)}_2)$ consists in the minimization of $\mathbf{\hat r_0}$ under the intersection of interval constraints. It is therefore straightforward to solve. In particular the optimal solution $\mathbf{\hat r_0^*}$ lies at the lower value of one of the intervals. Therefore there exists $\left( x^{(u_0),k^*_0},r^{(u_0),k^*_0},y^{(u_0), k^*_0} \right) \in \mathcal F^{(u_0)}$ such that
\begin{equation}
\mathbf{\hat r_0^*}= r^{(u_0), k^*_0} - L_{\rho} \left\|x_0-x^{(u_0), k^*_0} \right\|.
\label{solutionr0}
\end{equation}
Furthermore $\mathbf{\hat r_0^*}$ must belong to all intervals. We therefore have that
\begin{equation}
\mathbf{\hat r_0^*} \geq  r^{(u_0), k_0} - L_{\rho} \left\| x_0-x^{(u_0), k_0} \right\| \ , \qquad \forall k_0 \in \left\{ 1, \ldots , n^{(u_0)} \right\}. \label{intervalsr0}
\end{equation}
In other words,
\begin{eqnarray*}
\mathbf{\hat r_0^*} = \underset {k_0 \in \left\{ 1, \ldots , n^{(u_0)} \right\}  } {\max} r^{(u_0), k_0} - L_{\rho} \left\| x_0-x^{(u_0), k_0} \right\| \ . 
\end{eqnarray*}

\paragraph{About $( \mathcal P''^{(u_0,u_1)}_2 )$}
Again we observe that it is the minimization of $\mathbf{\hat r_1}$ under the intersection of interval constraints as well. The sizes of the intervals are however not fixed but determined by the variable $\mathbf{\hat x_1}$. If we denote the optimal solution of $( \mathcal P''^{(u_0,u_1)}_2 )$ by $\mathbf{\hat r_1^*}$ and $\mathbf{\hat x_1^*}$, we know that $\mathbf{\hat r_1^*}$ also lies at the lower value of one of the intervals. Hence there exists  $\left( x^{(u),k^*_1},r^{(u),k^*_1},y^{(u), k^*_1} \right) \in \mathcal F^{(u)}$  such that
\begin{equation}
\mathbf{\hat r_1^*} = r^{(u),k^*_1} - L_{\rho} \left\| \mathbf{\hat x_1^*} - x^{(u),k^*_1} \right\|.
\label{solutionr1}
\end{equation}
Furthermore $\mathbf{\hat r_1^*}$ must belong to all intervals. We therefore have that
\begin{equation}
\mathbf{\hat r_1^*} \geq  r^{(u),k_1} - L_{\rho} \left\| \mathbf{\hat x_1^*} - x^{(u),k_1} \right\| \ , \qquad \forall k_1 \in \left\{ 1 , \ldots, n^{(u)} \right\} .
\label{intervalsr1}
\end{equation}
We now discuss two cases depending on the sign of $\mathbf{\hat r_0^*} - \mathbf{\hat r_1^*}$.

\bigskip

$-$ \textbf{If} $\mathbf{\hat r_0^*} - \mathbf{\hat r_1^*} \geq 0$\\
Using \eqref{solutionr0} and \eqref{intervalsr1} with index $k^*_0$, we have
\begin{equation}
 \mathbf{\hat r^*_0} - \mathbf{\hat r^*_1}   \leq L_\rho \left(   \left\|  \mathbf{\hat x^*_1} - x^{(u), k^*_0}  \right\|  -  \left\|  x_0 - x^{(u), k^*_0}  \right\| \right)  \label{constNotUseful_discr}
\end{equation}
Since $\mathbf{\hat r_0^*} - \mathbf{\hat r_1^*} \geq 0$, we therefore have
\begin{eqnarray}
\left|  \mathbf{\hat r^*_0} - \mathbf{\hat r^*_1}  \right| \leq L_\rho \left(   \left\|  \mathbf{\hat x^*_1}  - x^{(u),k^*_0}  \right\|  -  \left\|  x_0 - x^{(u),k^*_0}  \right\| \right) \ . \label{constNotUseful_first_case_1}
\end{eqnarray}
Using the triangle inequality we can write
\begin{eqnarray}
\left \|  \mathbf{\hat x^*_1} - x^{(u),k^*_0} \right\| \leq  \left\|  \mathbf{\hat x^*_1} - x_0  \right\| +  \left\|  x_0 - x^{(u),k^*_0} \right\| \label{constNotUseful_first_case_2} \ .
\end{eqnarray}
Replacing \eqref{constNotUseful_first_case_2} in \eqref{constNotUseful_first_case_1} we obtain
\begin{equation*}
\left| \mathbf{\hat r^*_1} - \mathbf{\hat r^*_0}  \right| \leq L_\rho \left\|  \mathbf{\hat x^*_1} -  x_0  \right\|
\end{equation*}
which shows that $\mathbf{\hat r_0^*}$ and $\mathbf{\hat r_1^*}$ satisfy constraint \eqref{setConst_r_0_r_1_general}.  

\bigskip

$-$ \textbf{If} $\mathbf{\hat r_0^*} - \mathbf{\hat r_1^*} < 0$\\
Using \eqref{solutionr1} and \eqref{intervalsr0} with index $k^*_1$, we have
\begin{equation*}
\mathbf{\hat r^*_1} - \mathbf{\hat r^*_0} \leq L_\rho \left(  \left\| x_0 - x^{(u), k^*_1}  \right\|  -  \left\|  \mathbf{\hat x^*_1} - x^{(u), k^*_1}  \right\| \right)
\end{equation*}
and since $\mathbf{\hat r^*_0} - \mathbf{\hat r^*_1} < 0$,
\begin{eqnarray}
\left| \mathbf{\hat r^*_1} - \mathbf{\hat r^*_0} \right| \leq L_\rho \left(  \left\|  x_0 - x^{(u), k^*_1}  \right\| - \left\|  \mathbf{\hat x^*_1} - x^{(u), k^*_1}  \right\| \right). \label{constNotUseful_second_case_1}
\end{eqnarray}
Using the triangle inequality we can write 
\begin{eqnarray}
\left \|  x_0 - x^{(u), k^*_1} \right\| \leq  \left\|   x_0 - \mathbf{\hat x^*_1}  \right\| +  \left\|  \mathbf{\hat x^*_1} - x^{(u), k^*_1} \right\|. \label{constNotUseful_second_case_2}
\end{eqnarray}
Replacing \eqref{constNotUseful_second_case_2} in \eqref{constNotUseful_second_case_1} yields
\begin{eqnarray}
\left| \mathbf{\hat r^*_1} - \mathbf{\hat r^*_0} \right| \leq  L_\rho  \left\|   x_0 - \mathbf{\hat x^*_1}  \right\| \ , \nonumber
\end{eqnarray}
which again shows that $\mathbf{\hat r_0^*}$ and $\mathbf{\hat r_1^*}$ satisfy constraint \eqref{setConst_r_0_r_1_general}.

\noindent In both cases $\mathbf{\hat r_0^*} - \mathbf{\hat r_1^*} \geq 0$ and $\mathbf{\hat r_0^*} - \mathbf{\hat r_1^*} < 0$, we have shown that constraint \eqref{setConst_r_0_r_1_general} is satisfied. 
\end{proof}

In the following of the paper, we focus on the two subproblems $( \mathcal P'^{(u_0,u_1)}_2 )$ and  $( \mathcal P''^{(u_0,u_1)}_2 )$ rather than on $( \mathcal P^{(u_0,u_1)}_2 )$. From the proof of  Theorem \ref{theorem_decoupling} given above, we can directly obtain the solution of $( \mathcal P'^{(u_0,u_1)}_2 )$:

\begin{corollary}
\label{corollary:solution_P'2}
The solution of the problem $( \mathcal P'^{(u_0,u_1)}_2 )$ is
\begin{eqnarray*}
\mathbf{\hat r_0^*} = \underset {k_0 \in \left\{ 1, \ldots , n^{(u_0)} \right\}  } {\max} r^{(u_0), k_0} - L_{\rho} \left\| x_0-x^{(u_0), k_0} \right\| \ .
\end{eqnarray*}
\end{corollary}

\subsection{Complexity of $( \mathcal P''^{(u_0,u_1)}_2 )$}
\label{subsection:complexity}

The problem $(  \mathcal P'^{(u_0,u_1)}_2 )$ being solved, we now focus in this section on the resolution of $( \mathcal P''^{(u_0,u_1)}_2 )$. In particular, we show that it is NP-hard, even in the particular case where there is only one element in the sample $\mathcal F^{(u_1)} = \left\{ \left(  x^{(u_1),1} , r^{(u_1),1}, y^{(u_1),1}  \right)   \right\}$. In this particular case, the problem $( \mathcal P''^{(u_0,u_1)}_2 )$ amounts to maximizing of the distance $\left\| \mathbf{\hat x_1} - x^{(u_1),1} \right\|$ under an intersection of balls as we show in the following lemma.

\begin{lemma}
\label{lemma:specific_case}
If the cardinality of $\mathcal F^{(u_1)}$ is equal to $1$:
\begin{eqnarray*}
\mathcal F^{(u_1)} = \left\{ \left(  x^{(u_1),1} , r^{(u_1),1}, y^{(u_1),1}  \right)   \right\} \ ,
\end{eqnarray*}
then the optimal solution to $\left(\mathcal P''^{(u_0,u_1)}_2\right)$ satisfies 
\begin{eqnarray*}
\mathbf{\hat r_1^*}= r^{(u_1),1}-L_{\rho} \left\| \mathbf{\hat x_1^*} - x^{(u_1),1} \right\|
\end{eqnarray*}
where $\mathbf{\hat x_1^*}$ maximizes $\left\| \mathbf{\hat x_1} - x^{(u_1),1} \right\|$ subject to
\begin{eqnarray*}
\left\|\mathbf{\hat x_1}- y^{(u_0) , k_0} \right\|^2 \leq L^2_f \left\|x_0 - x^{(u_0),k_0} \right\|^2 \ , \qquad  \forall \left(x^{(u_0),k_0},r^{(u_0),k_0},y^{(u_0),k_0} \right)\in \mathcal F^{(u_0)} \ .
\end{eqnarray*}
\end{lemma}

\begin{proof}
The unique constraint concerning $\mathbf{\hat r_1}$ is an interval. Therefore $\mathbf{\hat r_1^*}$ takes the value of the lower bound of the interval. In order to obtain the lowest such value, the right-hand-side of \eqref{c1p2''} must be maximized under the other constraints.
\end{proof}

Note that if the cardinality $n^{(u_0)} $ of $\mathcal F^{(u_0)}$ is also equal to $1$, then $(  \mathcal P^{(u_0,u_1)}_2 )$ can be solved exactly, as we will later show in Corollary \ref{corollary:trivial_case_n0_n1}. But, in the general case where $n^{(u_0)} > 1 ,$ this problem of maximizing a distance under a set of ball-constraints is NP-hard as we now prove. To do it, we introduce the MNBC (for ``Max Norm with Ball Constraints") decision problem:

\begin{definition}[MNBC Decision Problem]
Given $x^{(0)}\in \mathbb Q^d, y^i\in \mathbb Q^d, \gamma_i\in \mathbb Q, i\in \{1,\ldots, I\}, C \in \mathbb Q$, the MNBC problem is to  determine whether there exists $x\in \mathbb R^d$ such that
$$\left\| x - x^{(0)} \right\|^2\geq C$$
and
$$\left\| x - y^i \right\|^2\leq \gamma_i \ , \qquad \forall i\in \{1,\ldots, I\} \ .$$
\end{definition}
\begin{lemma}
\label{lemma:MNBC}
MNBC is NP-hard.
\end{lemma}

\begin{proof}
To prove it, we will do a reduction from the $\{0,1\}-$programming feasibility problem \cite{Papadimitriou2003}. More precisely, we consider in this proof the  $\{0,2\}-$programming feasibility problem, which is equivalent. The problem is, given $p \in \mathbb N, A\in \mathbb Z^{p \times d}, b\in \mathbb Z^p$ to find whether there exists $x\in \{0,2\}^d$ that satisfies $Ax\leq b$. This problem is known to be NP-hard and we now provide a polynomial reduction to MNBC.

The dimension $d$ is kept the same in both problems. The first step is to define a set of constraints for MNBC such that the only potential feasible solutions are exactly $x\in \{0,2\}^d.$ We define $$x^{(0)} \triangleq (1,\ldots, 1)$$ and $$C \triangleq d .$$ For $i=1,\ldots, d$, we define $$y^{2i} \triangleq \left(y^{2i}_1,\ldots, y^{2i}_d\right)$$ with $y^{2i}_i \triangleq 0$ and $y^{2i}_j \triangleq 1$ for all $j\neq i$ and $\gamma_i \triangleq d+3$.\\
Similarly for $i=1,\ldots, d$, we  define $$y^{2i+1} \triangleq (y^{2i+1}_1,\ldots, y^{2i+1}_d)$$ with $y^{2i+1}_i \triangleq 2$ and $y^{2i+1}_j \triangleq 1$ for all $j\neq i$ and $\gamma_i \triangleq d+3$.\medskip

\noindent
\textbf{Claim}
$$\left\{x\in \mathbb R^d \mid \|x-x^{(0)}\|^2\geq d \right\}
\cap \left( \bigcap_{i=2}^{2d+1} \left\{x\in \mathbb R^d \mid  \|x-y^i\|^2\leq \gamma_i \right\}\right) = \{0,2\}^d$$
It is readily verified that any $x\in \{0,2\}^d$ belongs to the $2d+1$ above sets.\\
Consider $x\in \mathbb R^d$ that belongs to the $2d+1$ above sets. Consider an index $k\in\{1,\ldots, d\}$. Using the constraints defining the sets, we can in particular write
\begin{align*}
\|(x_1,\ldots, x_{k-1},x_k,x_{k+1},\ldots, x_d)-(1,\ldots, 1)\|^2 &\geq d\\
\|(x_1,\ldots, x_{k-1},x_k,x_{k+1},\ldots, x_d)-(1,\ldots,1,0,1,\ldots, 1)\|^2 &\leq d+3\\
\|(x_1,\ldots, x_{k-1},x_k,x_{k+1},\ldots, x_d)-(1,\ldots,1,2,1,\ldots, 1)\|^2 &\leq d+3
\end{align*}
that we can write algebraically
\begin{align}
\sum_{j\neq k} (x_j-1)^2 + (x_k-1)^2 &\geq d\label{geqn}\\
\sum_{j\neq k} (x_j-1)^2 + x_k^2 &\leq d+3\label{leqn1}\\
\sum_{j\neq k} (x_j-1)^2 + (x_k-2)^2 &\leq d+3.\label{leqn2}
\end{align}
By computing $\eqref{leqn1}-\eqref{geqn}$ and $\eqref{leqn2}-\eqref{geqn}$, we obtain $x_k\leq 2$ and $x_k\geq 0$ respectively. This implies that
$$\sum_{k=1}^d \left(x_k-1\right)^2\leq d$$
and the equality is obtained if and only if we have that $x_k\in \{0,2\}$ for all $k$ which proves the claim.\medskip

\noindent
It remains to prove that we can encode any linear inequality through a ball constraint. Consider an inequality of the type $\sum_{j=1}^d a_j x_j \leq b.$ We assume that $a\neq 0$ and that $b$ is even and therefore that there exists no $x\in \{0,2\}^d$ such that $a^Tx=b+1.$ We want to show that there exists $y\in \mathbb Q^d$ and $\gamma\in \mathbb Q$ such that
\begin{equation}
\left\{x\in \{0,2\}^d \mid a^Tx\leq b \right\} = \left\{x\in \{0,2\}^d \mid \|x-y\|^2\leq \gamma \right\}.
\label{theclaim}
\end{equation}
Let $\bar y\in \mathbb R^d$ be the intersection point of the hyperplane $a^Tx=b+1$ and the line $(1\ \cdots\ 1)^T+\lambda (a_1 \ \cdots \ a_d)^T, \lambda\in \mathbb R$. Let $r$ be defined as follows: $$r=\left\lceil \frac{d}{2} \sqrt{\sum_{j=1}^d a_j^2} + 1 \right\rceil.$$ We claim that choosing $\gamma \triangleq r^2$ and $y \triangleq \bar y - r a$ allows us to obtain \eqref{theclaim}. To prove it, we need to show that $x\in \{0,2\}^d$ belongs to the ball if and only if it satisfies the constraint $a^Tx\leq b$. Let $\bar x\in \{0,2\}^d$. There are two cases to consider:
\begin{itemize}
\item Suppose first that  $a^T\bar x\geq b+2$. 
\end{itemize}
Since $\bar y$ is the closest point to $y$ that satisfies $a^Ty=b+1$, it also implies that any point $x$ such that $a^Tx>b+1$ is such that $\|x-y\|^2>r^2$ proving that:
\begin{eqnarray*}
\bar{x} \notin \left\{x\in \mathbb R^d\mid \|x-y\|^2\leq r^2\right\}.
\end{eqnarray*}
\begin{itemize}
\item Suppose now that $a^T\bar x\leq b$ and in particular that $a^T\bar x=b-k$ with $k\in \mathbb N$ (see Figure \ref{fig-nphardproof}).
\end{itemize}
\begin{figure}
\begin{center}
\includegraphics[width=.45\linewidth]{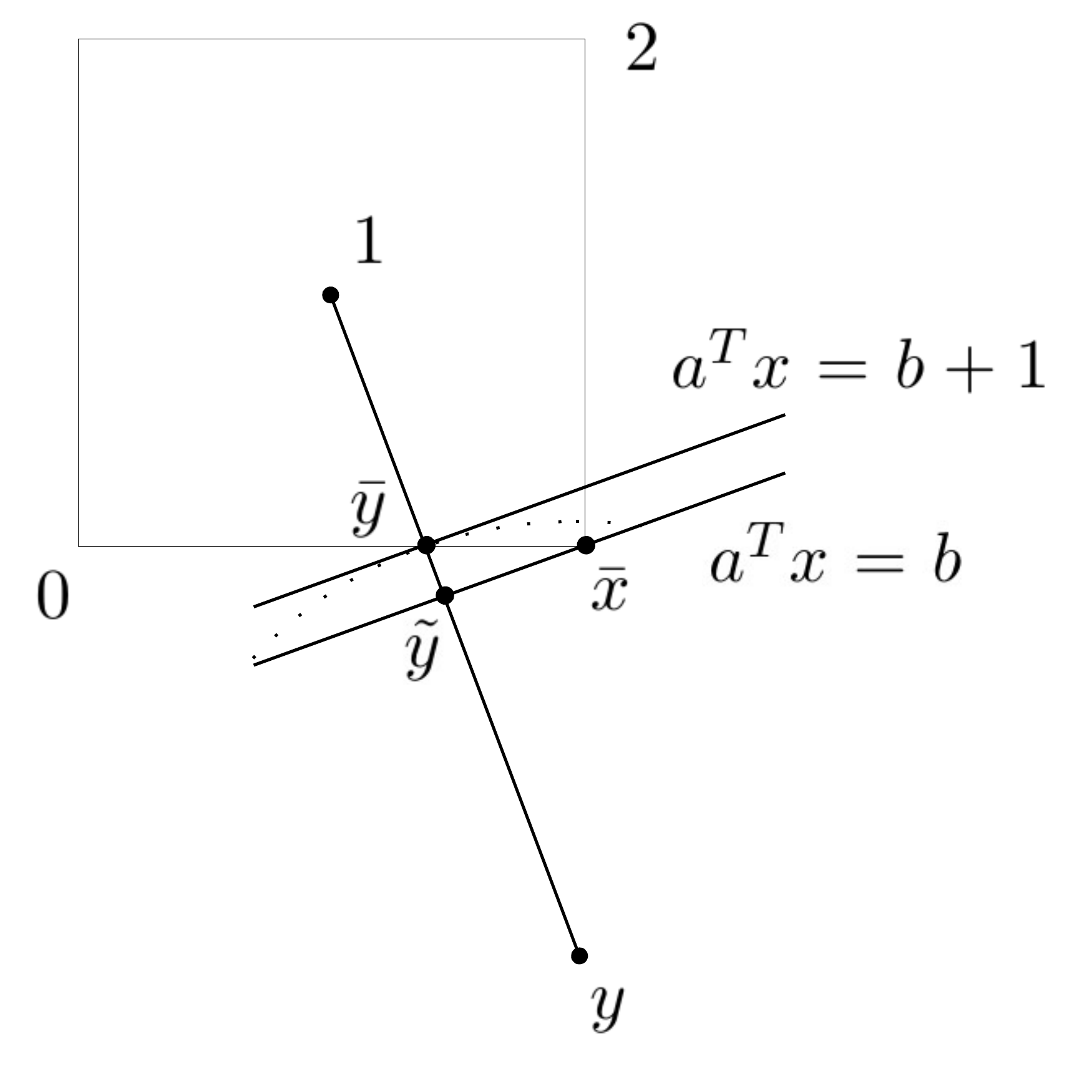}
\caption{The case when $a^T\bar x\leq b$.}
\label{fig-nphardproof}
\end{center}
\end{figure}
Let $\tilde y\in \mathbb R^d$ be the intersection point of the hyperplane $a^Tx=b-k$ and the line $(1\ \cdots\ 1)^T+\lambda (a_1 \ \cdots \ a_d)^T, \lambda\in \mathbb R$. Since $\left( (1\cdots 1)^T,\tilde y, \bar x\right)$ form a right triangle with the right angle in $\tilde y$ and since $\|(1\cdots 1)^T-\bar x\|^2\leq d$, we have
\begin{equation}
\|\tilde y - \bar x\|^2 \leq d.
\label{ytildexbar}
\end{equation}
By definition of $y$, we have:
$$\|y-\bar y\|=r \ ,$$
and by definition of $\bar y $ and $\tilde y$, we have:
$$\left\| \bar y-\tilde y \right\|\geq \frac{1}{\sqrt{\sum_{j=1}^d a_j^2}} \ .$$
Since $\bar y, \tilde y$ and $y$ belong to the same line, we have
\begin{equation}
\left\|y-\tilde y \right\|\leq r-\frac{1}{\sqrt{\sum_{j=1}^d a_j^2}}.
\label{yytilde}
\end{equation}
As $(y,\tilde y, \bar x)$ form a right triangle with the right angle in $\tilde y$, we have that
\begin{align*}
\|\bar x-y\|^2 &= \|y-\tilde y\|^2 + \|\bar x-\tilde y\|^2 \\
&\leq \left(r-\frac{1}{\sqrt{\sum_{j=1}^d a_j^2}} \right)^2 + d \quad \text{using }\eqref{ytildexbar},\eqref{yytilde}\\
&= r^2 -\frac{2r}{\sqrt{\sum_{j=1}^d a_j^2}}+\frac{1}{\sum_{j=1}^d a_j^2}+d.
\end{align*}
Since by definition, $r\geq \frac{d}{2}\sqrt{\sum_{j=1}^d a_j^2}+1$, we can write
\begin{align*}
\|\bar x-y\|^2&\leq r^2 -d-\frac{2}{\sqrt{\sum_{j=1}^d a_j^2}} + \frac{1}{\sum_{j=1}^d a_j^2}+d\\
&= r^2-\frac{1}{\sum_{j=1}^d a_j^2}\\
&\leq r^2.
\end{align*}
This proves that the chosen ball $\left\{x\in \mathbb R^d\mid \|x-y\|^2 \leq r^2 \right\}$ includes the same points from $\{0,2\}^d$ as the linear inequality $a^Tx\leq b$.

The encoding length of all data is furthermore polynomial in the encoding length of the initial inequalities. This completes the reduction and proves the NP-hardness of MNBC.
\end{proof}

 Note that the NP-hardness of MNBC is independent from the choice of the norm used over the state space $\mathcal X$. The two results follow:
\begin{corollary}
\label{corollary:complexity_second_stage}
$\left(\mathcal P''^{(u_0,u_1)}_2 \right) $ is NP-hard.
\end{corollary}
\begin{theorem}
\label{theorem:complexity}
The two-stage problem $\left(\mathcal P^{(u_0,u_1)}_2 \right)$  and the generalized $T-$stage problem $\left( \mathcal P_T(\mathcal F,L_f,L_\rho,x_0,u_0,\ldots,u_{T-1}) \right)$ are NP-hard.
\end{theorem}

\section{Relaxation Schemes for the Two-stage Case}
\label{section:relaxation_schemes_2_stage}

The two-stage case with only one element in the set $\mathcal F^{(u_1)}$ was proven to be NP-hard in the previous section (except if the cardinality of  $n^{(u_0)}$ of $\mathcal F^{(u_0)}$ is also equal to $1$,  in this case $(\mathcal P^{(u_0,u_1)}_2 )$ is solvable in polynomial time as we will see later in  Corollary \ref{corollary:trivial_case_n0_n1}). It is therefore unlikely that one can design an algorithm that optimally solves the general two-stage case in polynomial time (unless P = NP). The aim of the $\min \max$ optimization problem is to obtain a sequence of actions that has a performance guarantee. Therefore solving the optimization problem approximately or obtaining an upper bound would be irrelevant. Instead we want to propose some relaxation schemes that are computationally more tractable, and that are still leading to lower bounds on the actual return of the sequences of actions.

The first relaxation scheme works by dropping some constraints in order to obtain a problem that is solvable in polynomial time. We show that this scheme provides bounds that are greater or equal to the CGRL bound introduced in \cite{Fonteneau2011minmax}. The second relaxation scheme  is based on a Lagrangian relaxation where all constraints are dualized. Solving the Lagrangian dual is shown to be a conic-quadratic problem that can be solved in polynomial time using interior-point methods. We also prove that this relaxation scheme always gives better bounds than the first relaxation scheme mentioned above, and consequently, better bounds than \cite{Fonteneau2011minmax}. 
We also prove that the bounds computed from these relaxation schemes converge towards the actual return of the sequence $\left(u_0 , u_1 \right)$ when the sample dispersion converges towards zero. As a consequence, the sequences of actions that maximize those bounds also become optimal when the dispersion decreases towards zero.

From the previous section, we know that the two-stage problem $ ( \mathcal P^{(u_0,u_1)}_2 )$ can be decoupled into two subproblems $ ( \mathcal P'^{(u_0,u_1)}_2 )$ and $ ( \mathcal P''^{(u_0,u_1)}_2 )$, where $ ( \mathcal P'^{(u_0,u_1)}_2 )$ can be solved straightforwardly (cf Theorem \ref{theorem_decoupling}). We therefore only focus on relaxing the subproblem $ ( \mathcal P''^{(u_0,u_1)}_2 )$:
\begin{eqnarray}
&&\left( \mathcal P''^{(u_0,u_1)}_2 \right):  \underset {\begin{matrix} \mathbf{\hat r_{1}} \in \mathbb R\\
 \mathbf{\hat    x_{1}} \in \mathcal X\end{matrix}} {\min}\quad   \mathbf{\hat r_1} \nonumber \\
&&\text{subject to }\quad  \left| \mathbf{\hat r_1} - r^{(u_1),k_1} \right|^2 \leq L_\rho^2 \left\| \mathbf{\hat x_1} - x^{(u_1),k_1}  \right\|^2  \   \forall k_1 \in \left\{  1,\ldots, n^{(u_1)} \right\}  \label{rconstr} \\
&& \qquad \qquad \quad \left\| \mathbf{\hat x_{1}} - y^{(u_0),k_0}  \right\|^2 \leq L_f^2  \left\| x_{0} - x^{(u_0),k_0} \right\|^2 \qquad
  \forall k_0 \in \left\{ 1,\ldots, n^{(u_0)}  \right\} \label{xconstr} 
\end{eqnarray}

\subsection{The Trust-region Subproblem Relaxation Scheme}
\label{subsection:trust_region}

An easy way to obtain a relaxation from an optimization problem is  to drop some constraints. We therefore suggest to drop all constraints \eqref{rconstr} but one, indexed by $k_1$. Similarly we drop all constraints \eqref{xconstr} but one, indexed by $k_0$. The following problem is therefore a relaxation of $( \mathcal P''^{(u_0,u_1)}_2 )$:
\\
\, \fbox{
\begin{centering}
\begin{minipage}{0.95\textwidth}
$\left( \mathcal P''^{(u_0,u_1)}_{TR}(k_0,k_1) \right):$
\begin{eqnarray*}
\underset {\begin{matrix} \mathbf{ \hat r_{1}} \in \mathbb R\\
 \mathbf{\hat    x_{1}} \in \mathcal X\end{matrix}} {\min} \qquad  \mathbf{\hat r_1} \nonumber
\end{eqnarray*}
subject to
\begin{eqnarray}
&& \left| \mathbf{\hat r_1} - r^{(u_1),k_1} \right|^2 \leq L_\rho^2 \left\| \mathbf{\hat x_1} - x^{(u_1),k_1}  \right\|^2  \ ,  \label{rconstrTR} \\
&& \left\| \mathbf{\hat x_{1}} - y^{(u_0),k_0} \right\|^2 \leq L_f^2 \left\| x_{0} - x^{(u_0),k_0} \right\|^2   \ . \qquad   \label{xconstrTR} 
\end{eqnarray}
\end{minipage}
\end{centering}
}\\

We then have the following theorem:
\begin{theorem}
\label{theorem:trust_region}
Let us denote by $B''^{(u_0,u_1),k_0,k_1}_{TR}(\mathcal F)$ the bound given by the resolution of $( \mathcal P''^{(u_0,u_1)}_{TR}(k_0,k_1) )$. We have:
\begin{eqnarray*}
B''^{(u_0,u_1),k_0,k_1}_{TR}(\mathcal F) = r^{(u_1),k_1}-L_{\rho} \left\| \mathbf{\hat x_1^*}(k_0,k_1) -x^{(u_1),k_1} \right\|,\label{trbound_max2}
\end{eqnarray*}
where
\begin{eqnarray}
\mathbf{\hat x_1^*}(k_0,k_1) \doteq y^{(u_0),k_0} +  L_f\frac{ \left\| x_0 - x^{(u_0),k_0} \right\|}{\left\|y^{(u_0),k_0}-x^{(u_1),k_1} \right\|} \left(y^{(u_0),k_0}-x^{(u_1),k_1}\right)  \mbox{ if } y^{(u_0),k_0} \neq x^{(u_1),k_1} \nonumber
\end{eqnarray}
and, if $y^{(u_0),k_0} = x^{(u_1),k_1}$, $\mathbf{\hat x_1^*}(k_0,k_1)$ can be any point of the sphere centered in $y^{(u_0),k_0} = x^{(u_1),k_1}$ with radius $L_f \|  x_0 - x^{(u_0),k_0} \|$.
\end{theorem}
 
\begin{proof}
Observe that it consists in the minimization of $\mathbf{\hat r_1}$ under one interval constraint for $\mathbf{\hat r_1}$ where the size of the interval is determined through the constraint \eqref{xconstrTR}. The problem is therefore equivalent to finding the largest right-hand-side of \eqref{rconstrTR} under constraint \eqref{xconstrTR}. An equivalent problem is therefore
\begin{align*}
 \underset {\mathbf{\hat x_1} \in \mathcal X} {\max} \quad & \left\|\mathbf{\hat x_1} - x^{(u_1),k_1} \right\|^2  \\
\text{subject to } & \left\| \mathbf{\hat x_1} - y^{(u_0),k_0} \right\| \leq L_f \left\| x_0-x^{(u_0),k_0} \right\|.
\end{align*}
This is  the  maximization  of  a quadratic  function  under  a  norm constraint. This problem is referred to in the literature as the \emph{trust-region subproblem} \cite{Conn2000}.
In our case, the  optimal value for $\mathbf{\hat x_1}$  - denoted by $\mathbf{\hat x_1^*}(k_0,k_1)$ - lies on the same line  as $x^{(u_1),k_1}$ and $y^{(u_0),k_0}$,  with $y^{(u_0),k_0}$ lying in  between $x^{(u_1),k_1}$ and $\mathbf{\hat x_1^*}(k_0,k_1)$,  the distance between  $y^{(u_0),k_0}$ and $\mathbf{\hat x_1^*}(k_0,k_1)$  being exactly equal to the distance between $x_0$ and $x^{(u_0),k_0}$. An illustration is given in Figure \ref{figure:specific_case}.
\end{proof}
\begin{figure}[h!]
\begin{center}
\includegraphics[scale=0.35]{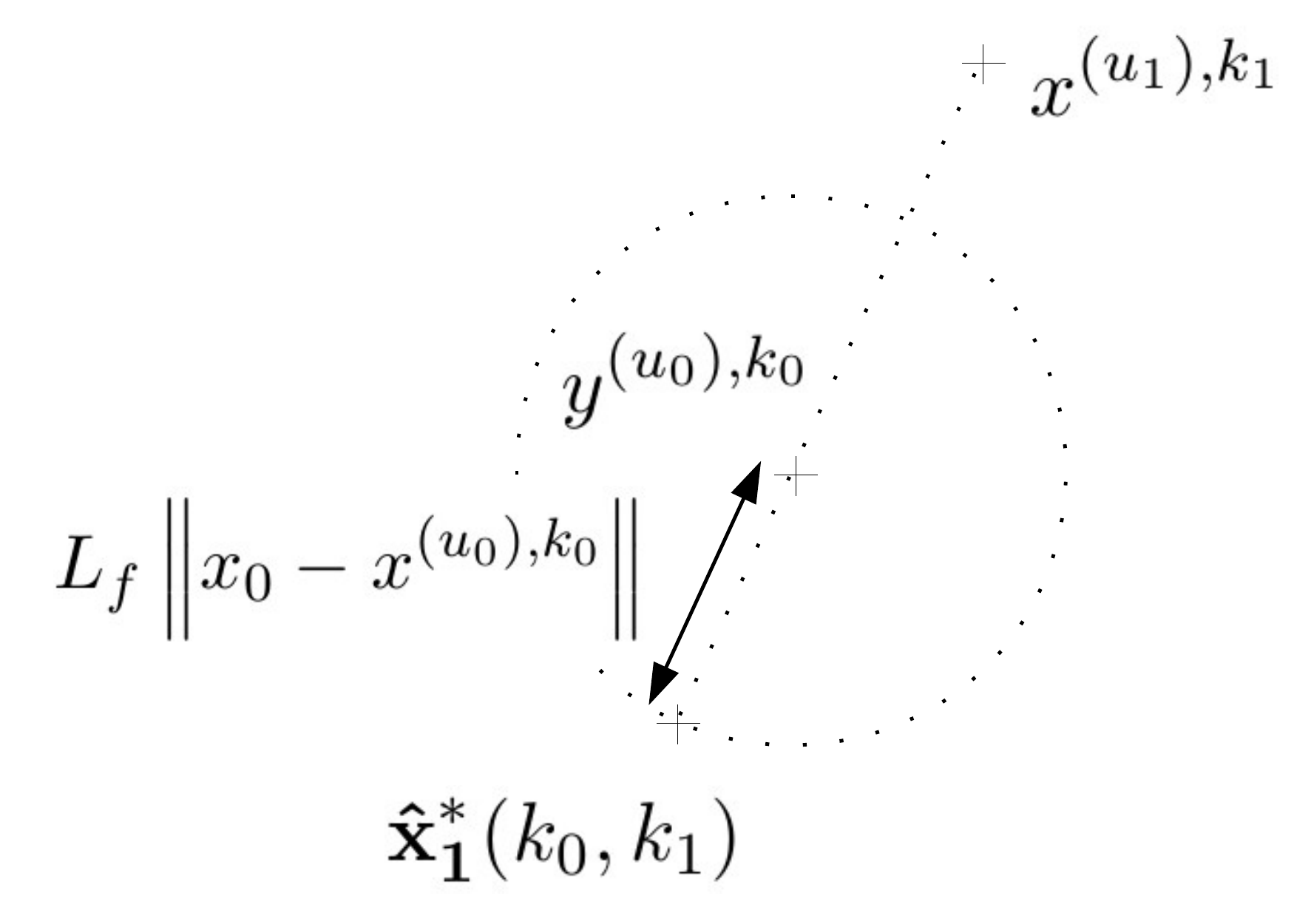}
\caption{A simple geometric algorithm to solve $(\mathcal P''_{TR}(k_0,k_1))$.  \label{figure:specific_case}}
\end{center}
\end{figure}

Solving $( \mathcal P''^{(u_0,u_1)}_{TR}(k_0,k_1) )$ provides us with a family of relaxations for our initial problem by considering any combination $(k_0 , k_1)$ of two non-relaxed constraints. 
Taking the maximum out of these lower bounds yields the best possible bound out of this family of relaxations. Finally, if we denote by $B^{(u_0,u_1)}_{TR}(\mathcal F)$ the bound made of the sum of the solution of $( \mathcal P'^{(u_0,u_1)}_2 )$ and the maximal Trust-region relaxation of the problem $( \mathcal P''^{(u_0,u_1)}_2 )$ over all possible couples of constraints, we have:
\begin{definition}[Trust-region Bound $B^{(u_0,u_1)}_{TR}\left( \mathcal F \right)$] 
\label{definition_TR}
\begin{eqnarray*}
\forall (u_0,u_1) \in \mathcal U^2, \quad B^{(u_0,u_1)}_{TR}(\mathcal F) \triangleq \mathbf{\hat r^*_0} +\underset {\begin{array}{l} k_1 \in \{1,\ldots, n^{(u_1)} \} \\k_0 \in \{1,\ldots,n^{(u_0)} \}\end{array}} {\max}  B''^{(u_0,u_1),k_0,k_1}_{TR}(\mathcal F) . 
\end{eqnarray*}
\end{definition}

Notice that in the case where $n^{(u_0)}$ and $n^{(u_1)}$ are both equal to $1$, then the trust-region relaxation scheme provides an exact solution of the original optimization problem $( \mathcal P^{(u_0,u_1)}_2 )$:
\begin{corollary}
\label{corollary:trivial_case_n0_n1}
\begin{eqnarray*}
\forall (u_0,u_1) \in \mathcal U^2, \qquad \left( \left\{ \begin{matrix} n^{(u_0)} =1 \\ n^{(u_1)} = 1 \end{matrix} \right. \right) \implies B^{(u_0,u_1)}_{TR}(\mathcal F)  = B^{(u_0,u_1)}_{2}(\mathcal F)  .
\end{eqnarray*}
\end{corollary}

\subsection{The Lagrangian Relaxation}
\label{subsection:lagrangian_relaxation}

Another way to obtain a lower bound on the value of a minimization problem is to consider a Lagrangian relaxation. In this section, we show that the Lagrangian relaxation of the second stage problem is a conic quadratic optimization program. Consider again the optimization problem $(\mathcal P''^{(u_0,u_1)}_2)$. If we multiply the constraints \eqref{rconstr} by dual variables $\mu_1 , \ldots ,  \mu_{k_1} , \ldots , \mu_{n^{(u_1)}} \geq 0$ and the constraints \eqref{xconstr} by dual variables $\lambda_1 , \ldots,  \lambda_{k_0} , \ldots, \lambda_{n^{(u_0)}} \geq 0$, we obtain the  Lagrangian dual:\\
\\
\fbox{
\begin{centering}
\begin{minipage}{0.95\textwidth}
$\left( \mathcal P''^{(u_0,u_1)}_{LD} \right) :$
\begin{align}
 \underset{ \begin{matrix} \lambda_1, \ldots , \lambda_{n^{(u_0)}} \in \mathbb R_+ \\ \mu_1 , \ldots, \mu_{n^{(u_1)}} \in \mathbb R_+ \end{matrix}}{\max} \quad &
\underset{ \begin{matrix} \mathbf{\hat r_1} \in \mathbb R \\ \mathbf{\hat x_1} \in \mathcal X \end{matrix}}{\min}\;
\qquad \mathbf{\hat r_1} \notag \\
&+\sum_{k_1=1}^{n^{(u_1)}} \mu_{k_1} \left(\left(\mathbf{\hat r_1} - r^{(u_1),k_1} \right)^2-L_{\rho}^2\left\| \mathbf{\hat x_1} - x^{(u_1),k_1}\right\|^2\right)\notag\\
&+\sum_{k_0=1}^{n^{(u_0)}} \lambda_{k_0} \left(\left \| \mathbf{\hat x_1} - y^{(u_0),k_0} \right\|^2-L_f^2\left\| x_0-x^{(u_0),k_0} \right\|^2\right)  \ . \label{lagdual}
\end{align}
\end{minipage}
\end{centering}
}\\

Observe that the optimal value of $(\mathcal P''^{(u_0,u_1)}_{LD})$ is known to provide a lower bound on the optimal value of $(  \mathcal P''^{(u_0,u_1)}_2 )$ \cite{LeMarechal1996}.
\begin{theorem} \label{theorem:conic_quadratic}
$\left(\mathcal P''^{(u_0,u_1)}_{LD} \right)$ is a conic quadratic program.
\end{theorem}

\begin{proof}
In \eqref{lagdual}, we can decompose the squared norms and obtain
\begin{align}
&\left( \mathcal P''^{(u_0,u_1)}_{LD} \right): \notag\\
&\underset{ \begin{matrix} \lambda_1 , \ldots , \lambda_{n^{(u_0)}} \in \mathbb R_+ \\ \mu_1 , \ldots, \mu_{n^{(u_1)}} \in \mathbb R_+ \end{matrix}}{\max} \quad \underset{\begin{matrix} \mathbf{\hat r_1} \in \mathbb R \\ \mathbf{\hat x_1} \in \mathcal X \end{matrix}}{\min}\;
\mathbf{\hat r_1}^2\left(\sum_{k_1=1}^{n^{(u_1)}} \mu_{k_1} \right) + \| \mathbf{\hat x_1}\|^2
\left(-L_{\rho}^2\sum_{k_1=1}^{n^{(u_1)}} \mu_{k_1} + \sum_{k_0=1}^{n^{(u_0)}} \lambda_{k_0} \right)\label{quadpart}\\
&+\mathbf{\hat r_1}\left(1-2\sum_{k_1=1}^{n^{(u_1)}} r^{(u_1),k_1}\right) + \sum_{k_1=1}^{n^{(u_1)}}2 L_{\rho}^2 \mu_{k_1} \left\langle \mathbf{\hat x_1},x^{(u_1),k_1} \right\rangle -\sum_{k_0=1}^{n^{(u_0)}}2 \lambda_{k_0} \left\langle \mathbf{\hat x_1} , y^{(u_0),k_0} \right\rangle  \label{linpart}\\
+  \sum_{k_1=1}^{n^{(u_1)}} & \mu_{k_1} \left( \left(r^{(u_1),k_1} \right)^2- L_{\rho}^2 \left\|x^{(u_1),k_1} \right\|^2 \right) \notag \\
&+\sum_{k_0=1}^{n^{(u_0)}} \lambda_{k_0} \left( \left\|y^{(u_0),k_0} \right\|^2 -L_f^2 \left\|x^{(u_0),k_0} - x_0 \right\|^2 \right), \label{constpart}
\end{align}
where $\langle a,b\rangle$ denotes the inner product of $a$ and $b$. We observe that the minimization problem in $\mathbf{\hat r_1}$ and $\mathbf{\hat x_1}$ contains a quadratic part \eqref{quadpart}, a linear part \eqref{linpart} and a constant part \eqref{constpart}
once we fix $\lambda_{k_0}$ and $\mu_{k_1}$. In particular, observe that the optimal solution of the minimization problem is $-\infty$ as soon as the quadratic term is negative, i.e. if :
\begin{eqnarray}
 \sum_{k_1=1}^{n^{(u_1)}} \mu_{k_1}  \leq 0 \label{relaxation_condition_1}
\end{eqnarray}
or
\begin{eqnarray}
\left(- L_{\rho}^2\sum_{k_1=1}^{n^{(u_1)}} \mu_{k_1} + \sum_{k_0=1}^{n^{(u_0)}} \lambda_{k_0} \right) \leq 0. \label{relaxation_condition_2}
\end{eqnarray}
Since we want to find the maximum of this series of optimization problems, we are only interested in the problems for which the solution is finite.
Observe that, since $\mu_{k_1} \geq 0$ for all $k_1$, the inequality (\ref{relaxation_condition_1}) is never satisfied, unless if $\mu_{k_1} = 0$ for all $k_1$. Therefore in the following, we will constraint $\lambda_{k_0}$ and $\mu_{k_1}$ to be such that inequalities (\ref{relaxation_condition_1}) and (\ref{relaxation_condition_2}) are never satisfied, i.e.:
\begin{eqnarray*}
 \sum_{k_1=1}^{n^{(u_1)}} \mu_{k_1}  &>& 0 \\  
-L_{\rho}^2\sum _{k_1=1}^{n^{(u_1)}} \mu_{k_1} + \sum_{k_0=1}^{n^{(u_0)}} \lambda_{k_0} &>& 0. \nonumber
\end{eqnarray*}
Once that  constraint is enforced, we observe that the minimization program is the minimization of a convex quadratic function for which the optimum can be found as a closed form formula. In order to simplify the rest of the proof, we introduce some useful notations:
\begin{definition}[Additional Notations]
\begin{eqnarray}
&&M \triangleq \sum_{k_1=1}^{n^{(u_1)}} \mu_{k_1} \quad , \quad  L \triangleq\sum_{k_0=1}^{n^{(u_0)}} \lambda_{k_0}  \ ,  \nonumber \\
&&X \triangleq\left(x^{(u_1),1} \cdots x^{(u_1),n^{(u_1)}}\right) \quad , \quad Y \triangleq \left(y^{(u_0),1}\cdots y^{(u_0),n^{(u_0)}}\right) \ ,  \nonumber \\
&&\llambda \triangleq \left(\begin{matrix} \lambda_1& \ldots & \lambda_{n^{(u_0)}} \end{matrix}\right)^T \quad , \quad  \mmu \triangleq \left(\begin{matrix} \mu_1 & \ldots & \mu_{n^{(u_1)}} \end{matrix}\right)^T \quad , \quad \bar{r} \triangleq \left(\begin{matrix} r^{(1)} & \ldots & r^{(n^{(u_1)})}\end{matrix}\right)^T \ ,     \nonumber \\
&&\forall p \in \mathbb N_0 ,  I_p\mbox{ is an identity matrix of size }p  .   \nonumber
\end{eqnarray} 
\end{definition}
The quadratic form coming from \eqref{quadpart}, \eqref{linpart} and \eqref{constpart} can be written in the form
\begin{eqnarray*}
 z^TQz+l^Tz+c
\end{eqnarray*}
with
\begin{equation*}
z \triangleq \left(\begin{array}{c}  \mathbf{\hat x_1} \\ \mathbf{\hat r_1} \end{array} \right)\in \R^{d+1} ,\quad \quad Q \triangleq \left(
\begin{array}{cc}
\left(-ML_{\rho}^2+L \right)I_d &\\&M 
\end{array}
\right),
\quad\quad
l \triangleq \left(
\begin{array}{c}
2L_{\rho}^2X\mmu-2Y\llambda\\
1-2\bar{r}^T\mmu
\end{array}
\right)
\end{equation*}
and the constant term is given by \eqref{constpart}. The minimum of a convex quadratic form $z^TQz+l^Tz+c$ is known to take the value $-\frac{1}{4} l^TQ^{-1}l+c.$ In our case, the inverse of the matrix $Q$ is trivial to compute and we obtain finally that $(\mathcal P''^{(u_0,u_1)}_{LD})$ can be written as
\begin{align}
\left( \mathcal P''^{(u_0,u_1)}_{LD} \right) :& \underset{\llambda\in \R^{n^{(u_0)}}_+, \mmu \in \R^{n^{(u_1)}}_+}{\max} \;
\frac{-\|L_{\rho}^2X\mmu-Y\llambda\|^2}{-ML_{\rho}^2+L}-\frac{\left(1-2\bar r^T\mmu\right)^2}{4M} \label{socp}\\
&+ \sum_{k_0=1}^{n^{(u_0)}} \lambda_{k_0} \left(\left\|y^{(u_0),k_0}\right\|^2 -L_f^2 \left\|x^{(u_0),k_0} - x_0 \right\|^2 \right) \notag \\
&+\sum_{k_1=1}^{n^{(u_1)}} \mu_{k_1} \left(\left(r^{(u_1),k_1}\right)^2- L_{\rho}^2 \left\|x^{(u_1),k_1}\right\|^2 \right) \notag
\end{align}
\begin{eqnarray*}
\mbox{ subject to } \qquad \qquad M &>& 0 \\
L &>& ML_{\rho}^2
\end{eqnarray*}

The optimization problem \eqref{socp} is in variables $\lambda_1, \ldots ,\lambda_{n^{(u_0)}}$ and $\mu_1, \ldots, \mu_{n^{(u_1)}}$. Observe that, with our notation, $M$ and $L$ are linear functions of the variables. The objective function contains linear terms in  $\lambda_1, \ldots ,\lambda_{n^{(u_0)}}$ and $\mu_1, \ldots, \mu_{n^{(u_1)}}$ as well as a \emph{fractional-quadratic} function (\cite{Ben2001}), i.e. the quotient of a concave quadratic function with a linear function. The constraint is linear. This type of problem is known as a \emph{rotated quadratic conic problem} and can be formulated as a conic quadratic optimization problem (\cite{Ben2001}) that can be solved in polynomial time using interior point methods \cite{Ben2001, Nesterov1994,Boyd2004}.
\end{proof}

 From there, we have the following corollary:

\begin{corollary}
$\forall (u_0,u_1) \in \mathcal U^2 $,
\begin{eqnarray}
B''^{(u_0,u_1)}_{LD}(\mathcal F) &\triangleq&  \underset{\llambda\in \R^{n^{(u_0)}}_+, \mmu \in \R^{n^{(u_1)}}_+}{\max} \;
\frac{-\|L_{\rho}^2X\mmu-Y\llambda\|^2}{-ML_{\rho}^2+L}-\frac{\left(1-2\bar r^T\mmu\right)^2}{4M}  \\
&& \quad \quad+ \sum_{k_0=1}^{n^{(u_0)}} \lambda_{k_0} \left(\left\|y^{(u_0),k_0}\right\|^2 -L_f^2 \left\|x^{(u_0),k_0} - x_0 \right\|^2 \right)  \\
&& \quad \quad +\sum_{k_1=1}^{n^{(u_1)}} \mu_{k_1} \left(\left(r^{(u_1),k_1}\right)^2- L_{\rho}^2 \left\|x^{(u_1),k_1}\right\|^2 \right) \notag
\end{eqnarray}
\begin{eqnarray*}
\mbox{ subject to } \qquad \qquad M &>& 0 \\
L &>& ML_{\rho}^2
\end{eqnarray*}
\end{corollary}
In the following, we denote by $B^{(u_0,u_1)}_{LD}(\mathcal F)$ the lower bound made of the sum of the solution of $( \mathcal P'^{(u_0,u_1)}_2 ) $ and the relaxation of $( \mathcal P''^{(u_0,u_1)}_2 )$ computed from the Lagrangian relaxation:
\begin{definition}[Lagrangian Relaxation Bound $B^{(u_0,u_1)}_{LD}(\mathcal F)$]
\begin{eqnarray}
\forall (u_0,u_1) \in \mathcal U^2, \qquad B^{(u_0,u_1)}_{LD}(\mathcal F) \triangleq \mathbf{\hat r^*_0} +  B''^{(u_0,u_1)}_{LD}(\mathcal F)
\end{eqnarray}
\end{definition}

\subsection{Comparing the Bounds}
\label{subsection:comparing_the_bounds}

The CGRL algorithm proposed in \cite{Fonteneau2010ICAART,Fonteneau2011minmax} for addressing the $\min \max$ problem uses the procedure described in \cite{Fonteneau2009ADPRL} for computing a lower bound on the return of a policy given a sample of trajectories.  More specifically, for a given sequence $\left(u_0, u_1 \right) \in \mathcal U^2 $, the program $\left( \mathcal P_T(\mathcal F,L_f,L_\rho,x_0,u_0,\ldots,u_{T-1}) \right)$ is replaced  by a lower bound $B^{(u_0,u_1)}_{CGRL}(\mathcal F)$. We may now wonder how this bound compares in the two-stage case with the two new bounds of $\left( \mathcal P^{(u_0,u_1)}_2 \right)  $ that we have proposed: the trust-region bound and the Lagrangian relaxation bound.

\subsubsection{Trust-region Versus CGRL}
\label{subsubsection:TRversusCGRL}
We first recall the definition of the CGRL bound in the two-stage case.
\begin{definition}[CGRL Bound $B^{(u_0,u_1)}_{CGRL}(\mathcal F)$]
\label{definition:CGRL}
$\forall (u_0, u_1) \in \mathcal U^2$,
\begin{eqnarray*}
B^{(u_0,u_1)}_{CGRL}(\mathcal F) \triangleq 
 \underset{\begin{array}{l} k_1 \in \{1,\ldots, n^{(u_1)} \}\\ k_0 \in \{1,\ldots,n^{(u_0)}\}\end{array}}{\max }
&&r^{(u_0),k_0}-L_{\rho}(1+L_f) \left\|x^{(u_0),k_0}- x_0 \right\| \\  
&&+ r^{(u_1),k_1}-L_{\rho} \left\|y^{(u_0),k_0}-x^{(u_1),k_1} \right\|.  \nonumber \\
\label{cgrlbound}
\end{eqnarray*}
\end{definition}
The following theorem shows that the Trust-region bound is always greater than or equal to the CGRL bound.
\begin{theorem}
\label{theorem:TRversusCGRL}
\begin{eqnarray*}
\forall ( u_0,u_1) \in \mathcal U^2 , \qquad B^{(u_0,u_1)}_{CGRL}(\mathcal F) \leq B^{(u_0,u_1)}_{TR}(\mathcal F)  \ .
\end{eqnarray*}
\end{theorem}

\begin{proof}
Let $k^*_0 \in \left\{ 1 , \ldots, n^{(u_0)} \right\}$ and $k^*_1 \in \left\{ 1 , \ldots, n^{(u_1)} \right\}$ be such that 
\begin{eqnarray}
B^{(u_0,u_1)}_{CGRL}(\mathcal F) =  r^{(u_0),k^*_0}-L_{\rho}(1+L_f) \left\|x^{(u_0),k^*_0}- x_0 \right\| 
+ r^{(u_1),k^*_1}-L_{\rho} \left\|y^{(u_0),k^*_0}-x^{(u_1),k^*_1} \right\|.  \nonumber 
\end{eqnarray}
Now, let us consider the solution $B''^{(u_0,u_1),k^*_0,k^*_1}_{TR}(\mathcal F)$ of the problem $( \mathcal P''^{(u_0,u_1)}_{TR}(k^*_0,k^*_1) )$, and let us denote by $B^{(u_0,u_1),k^*_0, k^*_1}$ the bound obtained if, in the definition of the value of $\hat r^*_0$ given in Corollary \ref{corollary:solution_P'2}, we fix the value of $k'_0$ to $k^*_0$ instead of maximizing over all possible $k'_0$:
\begin{eqnarray*}
B^{(u_0,u_1),k^*_0, k^*_1} = r^{(u_0),k^*_0} - L_\rho \left\| x_0 - x^{(u_0),k^*_0}  \right\| + B''^{(u_0,u_1),k^*_0,k^*_1}_{TR}(\mathcal F) 
\end{eqnarray*}
Since $ r^{(u_0),k^*_0} - L_\rho \left\| x_0 - x^{(u_0),k^*_0}  \right\| $ is smaller or equal to the solution $\mathbf{\hat r^*_0}$ of $( \mathcal P'^{(u_0,u_1)}_2 )$, one has:
\begin{eqnarray}
B_{TR}^{(u_0,u_1),k^*_0, k^*_1}(\mathcal F) \geq B^{(u_0,u_1),k^*_0, k^*_1} \ . \label{BTRgeqCGRL1}
\end{eqnarray}
Now, observe that:
\begin{eqnarray}
B^{(u_0,u_1),k^*_0, k^*_1} - B_{CGRL}^{(u_0,u_1)}(\mathcal F)&=&L_{\rho}L_f \left\|x^{(u_0),k^*_0} - x_0 \right\| +L_{\rho}\left\|y^{(u_0),k^*_0}-x^{(u_1),k^*_1}\right\| \nonumber \\
&-&L_{\rho} \left\| \mathbf{\hat x_1^*}(k^*_0,k^*_1) -x^{(u_1),k^*_1} \right\|  \ .
\label{strictcompare}
\end{eqnarray}
By construction, $\mathbf{\hat x_1^*}(k^*_0,k^*_1)$ lies on the same line as $y^{(u_0),k^*_0}$ and $x^{(u_1),k^*_1}$ (see Figure \ref{figure:specific_case}). Furthermore
\begin{equation}
\left\| \mathbf{\hat x_1^*}(k^*_0,k^*_1)  - x^{(u_1),k^*_1} \right\|=\left\| \mathbf{\hat x_1^*}(k^*_0,k^*_1) - y^{(u_0),k^*_0} \right\|+\left\|y^{(u_0),k^*_0}-x^{(u_1),k^*_1}\right\|.
\label{oneline}
\end{equation}
Using  \eqref{oneline} in \eqref{strictcompare} yields
\begin{eqnarray}
&&B^{(u_0,u_1),k^*_0, k^*_1} - B_{CGRL}^{(u_0,u_1),k^*_0, k^*_1}(\mathcal F)=L_{\rho}L_f \left\|x^{(u_0),k^*_0}- x_0 \right\|\notag\\
&&+L_{\rho}\left(\left\|y^{(u_0),k^*_0}-x^{(u_1),k^*_1}\right\|-\left\| \mathbf{\hat x_1^*}(k^*_0,k^*_1) - y^{(u_0),k^*_0} \right\|- \left\|y^{(u_0),k^*_0}-x^{(u_1),k^*_1} \right\|\right)\notag\\
&&= L_{\rho}\left( L_f \left\|x^{(u_0),k^*_0}- x_0 \right\|- \left\| \mathbf{\hat x_1^*}(k^*_0,k^*_1) - y^{(u_0),k^*_0} \right\|\right).\label{finaleq}
\end{eqnarray}
By construction, Equation \eqref{finaleq} is equal to 0 (see Figure \ref{figure:specific_case}), which proves the equality
of the two bounds:
\begin{eqnarray}
B^{(u_0,u_1),k^*_0, k^*_1}  = B_{CGRL}^{(u_0,u_1)}(\mathcal F)  \ . \label{BTRgeqCGRL2}
\end{eqnarray}
The final result is given by combining Equations (\ref{BTRgeqCGRL1}) and (\ref{BTRgeqCGRL2}).
\end{proof}

From the proof, one can observe that the gap between the CGRL bound and the Trust-region bound is only due to the resolution of $(  \mathcal P'^{(u_0,u_1)}_2 )$. Note that in the case where $k_0^*$ also belongs to the set $\underset {k_0 \in \{1 , \ldots, n^{(u_0)}  \}} {\arg \max} r^{(u_0),k_0} - L_\rho \left\| x^{(u_0),k_0} - x_0   \right\| $, then the bounds are equal. The two corollaries follow:
\begin{corollary}\label{corollary:TR_et_CGRL_super_close}
Let $ (u_0,u_1) \in \mathcal U^2$. Let $k^*_0 \in \left\{ 1 , \ldots, n^{(u_0)} \right\}$ and $k^*_1 \in \left\{ 1 , \ldots, n^{(u_1)} \right\}$ be such that:
\begin{eqnarray}
B^{(u_0,u_1)}_{CGRL}(\mathcal F) =  r^{(u_0),k^*_0}-L_{\rho}(1+L_f) \left\|x^{(u_0),k^*_0}- x_0 \right\| 
+ r^{(u_1),k^*_1}-L_{\rho} \left\|y^{(u_0),k^*_0}-x^{(u_1),k^*_1} \right\|.  \nonumber 
\end{eqnarray}
Then,
\begin{eqnarray*}
\left( k^*_0 \in \underset {k_0 \in \{1 , \ldots, n^{(u_0)}  \}} {\arg \max} r^{(u_0),k_0} - L_\rho \left\| x^{(u_0),k_0} - x_0   \right\| \right) \implies  B_{CGRL}^{(u_0,u_1)}(\mathcal F) =  B_{TR}^{(u_0,u_1)}(\mathcal F) \ .
\end{eqnarray*}
\end{corollary}
\begin{corollary}
\begin{eqnarray*}
\forall (u_0,u_1) \in \mathcal U^2,  \qquad \left(  n^{(u_0)} = 1 \right)  \implies B^{(u_0,u_1)}_{CGRL}(\mathcal F) = B^{(u_0,u_1)}_{TR}(\mathcal F) \ .
\end{eqnarray*}
\end{corollary}

\subsubsection{Lagrangian Relaxation Versus Trust-region}
\label{subsubsection:LDversusTR}
In this section, we prove that the lower bound obtained with the Lagrangian relaxation is always greater than or equal to the Trust-region  bound. To prove this result, we give a preliminary lemma:
\begin{lemma} \label{lemma:n0_n1_1}
Let $\left( u_0, u_1 \right) \in \mathcal U^2 $ and $\left(k_0,k_1\right) \in \left\{1, \ldots, n^{(u_0)} \right\} \times \left\{1, \ldots, n^{(u_1)} \right\}$. Consider again the problem $\left( \mathcal P''^{(u_0,u_1)}_{TR}(k_0,k_1) \right)$  where all constraints are dropped except the two defined by $\left(k_0,k_1\right)$:
\begin{alignat*}{3}
&\left( \mathcal P''^{(u_0,u_1)}_{TR}(k_0,k_1) \right):&
\underset
{\begin{matrix} \mathbf{\hat r_{1}} \in \mathbb R\\
 \mathbf{\hat  x_{1}} \in \mathcal X\end{matrix}} {\min}\quad &  \mathbf{\hat r_1} \nonumber \\
&&\text{subject to }\quad & \left| \mathbf{\hat r_1} - r^{(u_1),k_1} \right|^2 \leq L_\rho^2 \left\| \mathbf{\hat x_1} - x^{(u_1),k_1}  \right\|^2  \\
&&& \left\| \mathbf{\hat x_{1}} - y^{(u_0),k_0} \right\|^2 \leq L_f^2 \left\|  x_{0} - x^{(u_0),k_0} \right\|^2  \ . \nonumber
\end{alignat*}
Then,  the Lagrangian relaxation of $\left( \mathcal P''^{(u_0,u_1)}_{TR}(k_0,k_1) \right)$ leads to a bound denoted by\\ $B''^{(u_0,u_1),k_0,k_1}_ {LD}(\mathcal F) $ which is equal to the Trust-region bound $B''^{(u_0,u_1),k_0,k_1}_{TR}(\mathcal F)$, i.e.
\begin{eqnarray*}
B''^{(u_0,u_1),k_0,k_1}_ {LD}(\mathcal F) = B''^{(u_0,u_1),k_0,k_1}_{TR}(\mathcal F) \ .
\end{eqnarray*}
\end{lemma}

Proofs of this lemma can be found in \cite{Beck2007} and \cite{Bonnans2006}, but we also provide in  Appendix \ref{appendix:proof_n0_n1_1} a proof in our particular case. We then have the following theorem:

\begin{theorem}
\label{theorem:LDversusTR}
\begin{eqnarray*}
\forall \left( u_0,u_1 \right) \in \mathcal U^2, \qquad B^{(u_0,u_1)}_{TR}(\mathcal F) \leq B^{(u_0,u_1)}_{LD}(\mathcal F) .
\end{eqnarray*}
\end{theorem}

\begin{proof}
Let $\left( u_0,u_1 \right) \in \mathcal U^2$. Let $\left(k^*_0,k^*_1\right) \in \left\{1, \ldots, n^{(u_0)} \right\} \times \left\{1, \ldots, n^{(u_1)} \right\}$ be such that:
\begin{eqnarray*}
B^{(u_0,u_1)}_{TR}(\mathcal F) = \mathbf{ \hat r^*_0 }  +  B''^{(u_0,u_1),k^*_0,k^*_1}_{TR}(\mathcal F) .
\end{eqnarray*}
Considering $ \left(k_0,k_1\right) = \left(k^*_0,k^*_1\right)$ in Lemma \ref{lemma:n0_n1_1}, we have:
\begin{eqnarray}
B^{(u_0,u_1)}_{TR}(\mathcal F) = \mathbf{ \hat r^*_0 } +  B''^{(u_0,u_1),k^*_0,k^*_1}_{LD}(\mathcal F) \label{theorem:LDversusTR_1}
\end{eqnarray}

Then, one can observe that the Lagrangian relaxation of the problem $(  \mathcal P''^{(u_0,u_1)}_{TR}(k^*_0,k^*_1) )$ - from which $ B''^{(u_0,u_1),k^*_0,k^*_1}_{LD}(\mathcal F)$ is computed - is also a relaxation of the problem $( \mathcal P''^{(u_0,u_1)}_{LD}  ) $  for which all the dual variables corresponding to constraints that are not related with the system transitions $ \left( x^{(u_0),k^*_0},r^{(u_0),k^*_0},y^{(u_0),k^*_0}  \right)$ and $ \left( x^{(u_1),k^*_1},r^{(u_1),k^*_1},y^{(u_1),k^*_1}  \right) $ would be forced to zero, i.e.
\begin{eqnarray*}
B''^{(u_0,u_1),k^*_0,k^*_1}_{LD}(\mathcal F) &=& \underset{\llambda\in \R^{n^{(u_0)}}_+, \mmu \in \R^{n^{(u_1)}}_+}{\max} \; \frac{-\|L_{\rho}^2X\mmu-Y\llambda\|^2}{-ML_{\rho}^2+L}-\frac{\left(1-2\bar r^T\mmu\right)^2}{4M} \\
&&+ \sum_{k_0=1}^{n^{(u_0)}} \lambda_{k_0} \left(\left\|y^{(u_0),k_0}\right\|^2 -L_f^2 \left\|x^{(u_0),k_0} - \hat x_0 \right\|^2 \right)  \\
&&+\sum_{k_1=1}^{n^{(u_1)}} \mu_{k_1} \left(\left(r^{(u_1),k_1}\right)^2- L_{\rho}^2 \left\|x^{(u_1),k_1}\right\|^2 \right) \\
&&\text{subject to } \qquad M > 0 \ , \\ 
&& \qquad \quad \quad \qquad L> ML_{\rho}^2 \ , \\
&&\qquad \quad \quad \qquad \lambda_{k_0} = 0 \mbox{ if } k_0 \neq k^*_0,  \forall k_0 \in \left\{1, \ldots , n^{(u_0)} \right\} \ , \\
&&\qquad \quad \quad \qquad \mu_{k_1} = 0 \mbox{ if } k_1 \neq k^*_1, \forall k_1 \in \left\{1, \ldots , n^{(u_1)} \right\} \ .
\end{eqnarray*}
We therefore have:
\begin{eqnarray}
B''^{(u_0,u_1),k^*_0,k^*_1}_{LD}(\mathcal F) \leq B''^{(u_0,u_1)}_{LD}(\mathcal F) \ . \label{theorem:LDversusTR_2}
\end{eqnarray}
By definition of the Lagrangian relaxation bound $B^{(u_0,u_1)}_{LD}(\mathcal F)$, we have:
\begin{eqnarray}
B^{(u_0,u_1)}_{LD}(\mathcal F) = \mathbf{\hat r^*_0} + B''^{(u_0,u_1)}_{LD}(\mathcal F) \label{theorem:LDversusTR_3}
\end{eqnarray}
Equations (\ref{theorem:LDversusTR_1}), (\ref{theorem:LDversusTR_2})  and (\ref{theorem:LDversusTR_3}) finally give:
\begin{eqnarray*}
B^{(u_0,u_1)}_{TR}(\mathcal F') = B^{(u_0,u_1)}_{LD}(\mathcal F) \ . 
\end{eqnarray*}
\end{proof}


\subsubsection{Bounds Inequalities: Summary}
\label{subsubsection:synthesis}
We summarize in the following theorem all the results that were obtained in the previous sections.
\begin{theorem}
\label{theorem:synthesis}
$\forall \left( u_0, u_1 \right) \in \mathcal U^2, $
\begin{eqnarray*}
B^{(u_0,u_1)}_{CGRL}(\mathcal F) \leq B^{(u_0,u_1)}_{TR}(\mathcal F) \leq B^{(u_0,u_1)}_{LD}(\mathcal F) \leq B^{(u_0,u_1)}_{2}(\mathcal F) \leq J^{(u_0,u_1)}_2 \ .
\end{eqnarray*}
\end{theorem}

\begin{proof}
Let $\left( u_0, u_1 \right) \in \mathcal U^2 $. The inequality 
\begin{eqnarray}
B^{(u_0,u_1)}_{CGRL}(\mathcal F) \leq B^{(u_0,u_1)}_{TR}(\mathcal F) \leq B^{(u_0,u_1)}_{LD}(\mathcal F)
\end{eqnarray}
is a straightforward consequence of Theorems \ref{theorem:TRversusCGRL} and \ref{theorem:LDversusTR}. The inequality
\begin{eqnarray}
 B^{(u_0,u_1)}_{LD}(\mathcal F) \leq B^{(u_0,u_1)}_{2}(\mathcal F)
\end{eqnarray}
is a property of the Lagrangian relaxation, and the inequality
\begin{eqnarray*}
B^{(u_0,u_1)}_{2}(\mathcal F) \leq J^{(u_0,u_1)}_2
\end{eqnarray*}
is derived from the formalization of the $\min \max$ generalization problem introduced in \cite{Fonteneau2011minmax}.
\end{proof}

\subsection{Convergence Properties}
\label{subsection:convergence_bounds}
We finally propose to analyze the convergence of the bounds, as well as the sequences of actions that lead to the maximization of the bounds, when the sample dispersion decreases towards zero. We assume in this section that the state space $\mathcal X$ is bounded:
\begin{eqnarray*}
\exists C_{\mathcal X} > 0 : \forall (x,x') \in \mathcal X^2, \qquad \| x - x' \| \leq C_{\mathcal X} \ .
\end{eqnarray*}
Let us now introduce the sample dispersion:
\begin{definition}[Sample Dispersion]
Since $\mathcal X $ is bounded, one has:
\begin{eqnarray}
\exists \ \alpha > 0 : \forall u \in \mathcal U  \ , \qquad  \underset { x \in \mathcal X } {\sup} \quad \quad \underset { k \in \left\{ 1 , \ldots, n^{(u)}   \right\}  }
{\min} \left\| x^{(u),k} - x \right\|    \leq \alpha \ . \label{Eqn:hypo_alpha} 
\end{eqnarray}
The smallest $\alpha$  which satisfies equation (\ref{Eqn:hypo_alpha}) is  named  the sample  dispersion  and  is  denoted by  $\alpha^*(\mathcal F)$.
\end{definition}
Intuitively, the   sample dispersion $\alpha^*(\mathcal F)$ can be seen as the radius of the largest non-visited state space area.

\subsubsection{Bounds}

We  analyze in this    subsection    the     tightness    of    the Trust-region and the Lagrangian relaxation  lower    bounds   as   a    function   of   the   sample dispersion. 
\begin{lemma}
\label{lemma:tightness}$\ $\\
$\exists  \  C  >  0 :   \forall      (u_0,u_1)    \in     \mathcal      U^2, \forall B^{(u_0,u_1)}(\mathcal F) \in \left\{  B^{(u_0,u_1)}_{CGRL}(\mathcal F) , B^{(u_0,u_1)}_{TR}(\mathcal F) , B^{(u_0,u_1)}_{LD}(\mathcal F) \right\},$
\begin{eqnarray*}
J^{(u_0,u_1)}_2 - B^{(u_0,u_1)}(\mathcal F) \leq C  \alpha^*(\mathcal F).
\end{eqnarray*}
\end{lemma}

\begin{proof}
The proof for the case where $B^{(u_0,u_1)}(\mathcal F) = B^{(u_0,u_1)}_{CGRL}(\mathcal F)$ is given in \cite{Fonteneau2010ICAART}. According to Theorem \ref{theorem:synthesis}, one has:
\begin{eqnarray*}
\forall \left( u_0,u_1 \right) \in \mathcal U^2, \qquad B^{(u_0,u_1)}_{CGRL}(\mathcal F)  \leq B^{(u_0,u_1)}_{TR}(\mathcal F) \leq B^{(u_0,u_1)}_{LD}(\mathcal F)  \leq J^{(u_0,u_1)}_2 ,
\end{eqnarray*}
which ends the proof.
\end{proof}

We therefore have the following theorem:
\begin{theorem}
\label{theorem:consistency_bounds}
$ \\ $
$\forall      (u_0,u_1)    \in     \mathcal      U^2, \forall B^{(u_0,u1)}(\mathcal F) \in \left\{  B^{(u_0,u_1)}_{CGRL}(\mathcal F) , B^{(u_0,u_1)}_{TR}(\mathcal F) , B^{(u_0,u_1)}_{LD}(\mathcal F) \right\},$
\begin{eqnarray*}
\underset {\alpha^*(\mathcal F)  \rightarrow 0} {\lim} \quad \quad J^{(u_0,u_1)}_2 - B^{(u_0,u_1)}(\mathcal F) = 0 \ .
\end{eqnarray*}
\end{theorem}

\subsubsection{Bound-optimal Sequences of Actions}
\label{subsection:convergence_actions}

In the following, we denote by $B^{(*)}_{CGRL}\left(\mathcal F \right)$  (resp. $B^{(*)}_{TR}\left(\mathcal F \right)$ and $B^{(*)}_{LD}\left(\mathcal F  \right)$ ) the maximal CGRL bound  (resp. the maximal Trust-region bound and maximal Lagrangian relaxation bound) over the set of all possible sequences of actions, i.e.,
\begin{definition}[Maximal Bounds]
\begin{eqnarray*}
B^{(*)}_{CGRL}\left(\mathcal F \right) &\triangleq& \underset {(u_0,u_1) \in \mathcal U^2} {\max} B^{(u_0,u_1)}_{CGRL}\left(\mathcal F \right) \ , \nonumber \\
B^{(*)}_{TR}\left(\mathcal F \right) &\triangleq& \underset {(u_0,u_1) \in \mathcal U^2} {\max} B^{(u_0,u_1)}_{TR}\left(\mathcal F \right) \ , \nonumber \\
B^{(*)}_{LD}\left(\mathcal F \right) &\triangleq& \underset {(u_0,u_1) \in \mathcal U^2} {\max} B^{(u_0,u_1)}_{LD}\left(\mathcal F \right) \ . \nonumber
\end{eqnarray*}
\end{definition}
We also denote by $\left(u_0 , u_1  \right)^{CGRL}_{\mathcal F }$ (resp. $\left(u_0 , u_1  \right)^{TR}_{\mathcal F }$  and $\left(u_0 , u_1  \right)^{LD}_{\mathcal F}$) three sequences of actions that maximize the bounds:
\begin{definition}[Bound-optimal Sequences of Actions]
\begin{eqnarray*}
\left(u_0 , u_1  \right)^{CGRL}_{\mathcal F} &\in&   \left\{ (u_0,u_1) \in \mathcal U^2  |  B^{(u_0,u_1)}_{CGRL}  \left(\mathcal F \right) = B^{(*)}_{CGRL}\left(\mathcal F \right)   \right\}  \nonumber \\
\left(u_0 , u_1  \right)^{TR}_{\mathcal F} &\in&   \left\{ (u_0,u_1) \in \mathcal U^2  |  B^{(u_0,u_1)}_{TR}  \left(\mathcal F \right) = B^{(*)}_{TR}\left(\mathcal F \right)   \right\}  \nonumber \\
\left(u_0 , u_1  \right)^{LD}_{\mathcal F} &\in&   \left\{ (u_0,u_1) \in \mathcal U^2  |  B^{(u_0,u_1)}_{LD}  \left(\mathcal F \right) = B^{(*)}_{LD}\left(\mathcal F \right)   \right\}  \nonumber 
\end{eqnarray*}
\end{definition}

We finally give in this section a last theorem that shows the convergence of the sequences of actions $\left(u_0 , u_1  \right)^{CGRL}_{\mathcal F }$, $\left(u_0 , u_1  \right)^{TR}_{\mathcal F }$  and $\left(u_0 , u_1  \right)^{LD}_{\mathcal F}$  towards optimal sequences of actions - i.e. sequences of actions that lead to an optimal return $J^*_2$ - when the sample dispersion $\alpha^*(\mathcal F)$ decreases towards zero.

\begin{theorem}
\label{theorem_convergence_actions}  
Let $\mathfrak  J^*$ be the set of optimal two-stage sequences of actions:
\begin{eqnarray*}
\mathfrak  J^*_2 \triangleq  \left\{  \left(u_0,u_1 \right) \in  \mathcal  U^2   | J^{\left( u_0,u_1 \right) }_2 = J^*_2 \right\} \ ,
\end{eqnarray*}
 and let us suppose that   $\mathfrak J^*_2 \neq \mathcal  U^2$ (if $\mathfrak J^*_2 = \mathcal  U^2$,  the search  for  an  optimal sequence  of  actions is  indeed   trivial).   We define
\begin{eqnarray*}
\epsilon \triangleq  \underset {(u_0,u_1)   \in \mathcal U^2  \backslash \mathfrak J^*_2 } {\min}  \left\{ J^*_2 -   J^{(u_0,u_1)}_2 \right \}  \ .
\end{eqnarray*}
Then, 
$\forall \left( \tilde u_0, \tilde u_1 \right)_{\mathcal F} \in \left\{  \left(u_0 , u_1  \right)^{CGRL}_{\mathcal F } , \left(u_0 , u_1  \right)^{TR}_{\mathcal F } ,\left(u_0 , u_1  \right)^{LD}_{\mathcal F}  \right\} $,
\begin{eqnarray}
\Big( C  \alpha^*(\mathcal F)   <  \epsilon \Big)  \implies \left( \tilde u_0, \tilde u_1 \right)_{\mathcal F} \in  \mathfrak J^*_2 \ .
\end{eqnarray}
\end{theorem}

\begin{proof}
Let us prove the theorem by contradiction. Let us assume that $C  \alpha^*(\mathcal F)  <  \epsilon$. Let $B^{(u_0,u1)}(\mathcal F) \in \left\{  B^{(u_0,u_1)}_{CGRL}(\mathcal F) , B^{(u_0,u_1)}_{TR}(\mathcal F) , B^{(u_0,u_1)}_{LD}(\mathcal F) \right\}$, and let $(\tilde u_0 , \tilde u_1)_{\mathcal F}$ be a sequence such that
\begin{eqnarray*}
(\tilde u_0 , \tilde u_1)_{\mathcal F} \in  \underset {(u_0,u_1)   \in \mathcal U^2} {\arg\max}  B^{(u_0,u_1)}(\mathcal F)
\end{eqnarray*}
and let us assume that $(\tilde u_0 , \tilde u_1)_{\mathcal F}$ is not optimal. This implies that
\begin{eqnarray*}
J^{(\tilde u_0, \tilde u_1)_{\mathcal F}}_2 \leq J^*_2  -    \epsilon   \    .   
\end{eqnarray*}
Now, let   us    consider    a   sequence $\left( u^*_0, u^*_{1} \right)    \in     \mathfrak   J^*_2$.     Then
\begin{eqnarray*}
J^{( u^*_0, u^*_{1})}_2  =  J^*_2 . 
\end{eqnarray*}
The  lower  bound  $B^{ (u^*_0, u^*_{1})}(\mathcal F)$   satisfies   the   relationship
\begin{eqnarray*}
J^*_2 - B^{(u^*_0, u^*_{1})}(\mathcal F) \leq C \alpha^*(\mathcal F) .
\end{eqnarray*}
Knowing that     $C      \alpha^*(\mathcal F)     <     \epsilon$,      we     have
\begin{eqnarray*}
B^{(u^*_0, u^*_{1})}(\mathcal F)   >  J^*_2   -  \epsilon   .
\end{eqnarray*}
Since
\begin{eqnarray*}
J^{(\tilde u_0, \tilde u_1)_{\mathcal F}}_2   \geq    B^{(\tilde u_0, \tilde u_1)_{\mathcal F}}(\mathcal F) ,  
\end{eqnarray*}
we    have
\begin{eqnarray*}
B^{(u^*_0, u^*_1)}(\mathcal F)  >  B^{(\tilde u_0, \tilde  u_1)_{\mathcal F}}(\mathcal F)
\end{eqnarray*}
which contradicts the fact that $(\tilde u_0 , \tilde u_1)_{\mathcal F}$ belongs to the set $\underset {(u_0,u_1)   \in \mathcal U^2} {\arg\max}  B^{(u_0,u_1)}(\mathcal F)$. This ends the proof.
\end{proof}

\subsubsection{Remark}

It is important to notice that the tightness of the  bounds resulting from  the relaxation schemes proposed in this paper does   not depend  {\it explicitly} on  the sample  dispersion (which suffers from the curse of dimensionality), but  depends rather on the initial  state for which the sequence of actions is  computed and on the local concentration of samples around the actual (unknown) trajectories of the system.  Therefore,  this  may  lead to  cases  where  the  bounds are tight  for  some  specific initial states,  even if the sample  does not cover every  area of the state space well enough.

\section{Experimental Results}
\label{section:experimental_results}

We provide some experimental results to illustrate the theoretical properties of the CGRL, Trust-region and Lagrangian relaxation bounds given below. We compare the tightness of the bounds, as well as the performances of the bound-optimal sequences of actions, on an academic benchmark.

\subsection{Benchmark}
\label{subsection:benchmark}
We  consider a linear benchmark whose dynamics is defined as follows :
\begin{eqnarray*}
\forall (x,u) \in \mathcal X \times \mathcal U, \qquad f(x,u) = x + 3.1416 \times u \times 1_d ,
\end{eqnarray*}
where $1_d \in \mathbb R^d$ denotes a $d-$dimensional vector for which each component is equal to $1$. The reward function is defined as follows:
\begin{eqnarray*}
\forall (x,u) \in \mathcal X \times \mathcal U, \qquad \rho(x,u) =  \sum_{i=1}^{d} x(i) \ ,
\end{eqnarray*}
where $x(i)$ denotes the $i-$th component of $x$. The state space $\mathcal X$ is included in $\mathbb R^{d}$ and the finite action space is equal to $\mathcal U = \{0 , 0.1\}$. The system dynamics $f$ is $1-$Lipschitz continuous and the reward function is $\sqrt{d}-$Lipschitz continuous. The initial state of the system is set to $$x_0 = 0.5772 \times 1_d \ . $$ The dimension $d$ of the state space is set to $d = 2$. In all our experiments, the computation of the Lagrangian relaxations, which requires to solve a conic-quadratic program,  are done using SeDuMi \cite{Sturm1999}.

\subsection{Protocol and Results}
\label{subsection:results}
\begin{figure}[h!]
\begin{center}
\includegraphics[scale=.75]{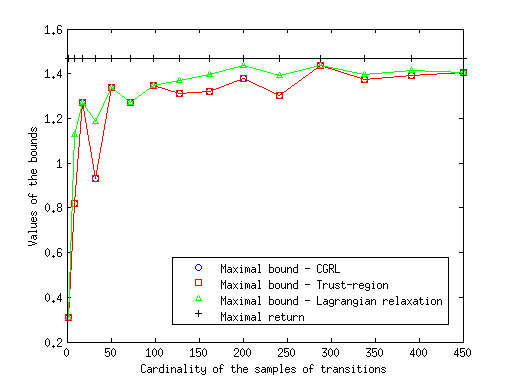}
\caption{Bounds $ B^{(*)}_{CGRL}\left(\mathcal F_{c_i} \right)$, $ B^{(*)}_{TR}\left(\mathcal F_{c_i} \right)$ and $ B^{(*)}_{LD}\left(\mathcal F_{c_i} \right) $ computed from all samples of transitions $\mathcal F_{c_i} \quad i \in \{1, \ldots, 15\}$ of cardinality $c_i = 2 i^2 .$ \label{figure:linear_cardinality_grid_bounds}}
\end{center}
\end{figure}
\begin{figure}[h!]
\begin{center}
\includegraphics[scale=.75]{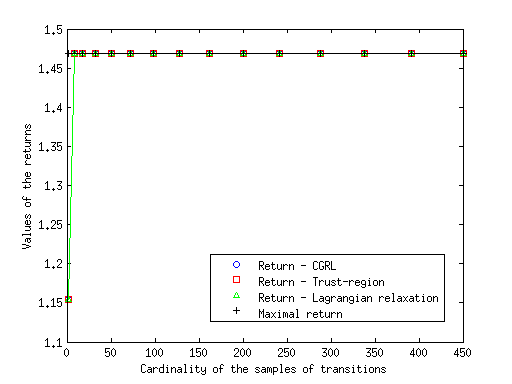}
\caption{Returns of the sequences  $\left(u_0 , u_1  \right)^{CGRL}_{\mathcal F_{c_i} }$, $\left(u_0 , u_1  \right)^{TR}_{\mathcal F_{c_i} }$ and $\left(u_0 , u_1  \right)^{LD}_{\mathcal F_{c_i}}$ computed from all samples of transitions $\mathcal F_{c_i} \quad i \in \{1, \ldots, 15\}$ of cardinality $c_i = 2 i^2 .$ \label{figure:linear_cardinality_grid_returns}}
\end{center}
\end{figure}

\subsubsection{Typical Run} 
\label{subsubsection:typical_run}

For different cardinalities $c_i =2 i^2 , i=1 , \ldots , 15$, we generate a sample of transitions $\mathcal F_{c_i} $ using a grid over $[ 0 , 1]^{d} \times \mathcal U$, as follows: $\forall u \in \mathcal U,$
\begin{eqnarray*}
 \mathcal F^{(u)}_{c_i} = \left\{  \left( \left[ \frac{i_1}{i} ; \frac{i_2}{i} \right] , u , \rho\left( \left[ \frac{i_1}{i} ; \frac{i_2}{i} \right] ,u \right) , f\left(\left[ \frac{i_1}{i} ; \frac{i_2}{i} \right],u\right) \right) \bigg| (i_1,i_2) \in \{1, \ldots, i\}^2 \right\} 
\end{eqnarray*}
and
\begin{eqnarray*}
\mathcal F_{c_i} = \mathcal F^{(0)}_{c_i} \cup \mathcal F^{(.1)}_{c_i}
\end{eqnarray*}
We report in Figure \ref{figure:linear_cardinality_grid_bounds} the values of the maximal CGRL bound $ B^{(*)}_{CGRL}\left(\mathcal F_{c_i}\right)$, the maximal Trust-region bound $ B^{(*)}_{TR}\left(\mathcal F_{c_i}\right)$ and the maximal Lagrangian relaxation bound $B^{(*)}_{LD}\left(\mathcal F_{c_i}\right)$ as a function of the cardinality $c_i$ of the samples of transitions $\mathcal F_{c_i} $. We also report in Figure \ref{figure:linear_cardinality_grid_returns} the returns  $J^{\left(u_0 , u_1  \right)^{CGRL}_{\mathcal F_{c_i}}}_2$, $J^{\left(u_0 , u_1  \right)^{TR}_{\mathcal F_{c_i}}}_2$ and $J^{\left(u_0 , u_1  \right)^{LD}_{\mathcal F_{c_i}}}_2$ of the bound-optimal sequences of actions  $\left(u_0 , u_1  \right)^{CGRL}_{\mathcal F_{c_i} }$, $\left(u_0 , u_1  \right)^{TR}_{\mathcal F_{c_i} }$ and $\left(u_0 , u_1  \right)^{LD}_{\mathcal F_{c_i}}$.

As expected, we observe that the bound computed with the Lagrangian relaxation is always greater or equal to the Trust-region bound, which is also greater or equal to the CGRL bound as predicted by Theorem \ref{theorem:synthesis}. On the other hand, no difference were observed in terms of return of the bound-optimal sequences of actions.

\subsubsection{Uniformly Drawn Samples of Transitions}

\begin{figure}[h!]
\begin{center}
\includegraphics[scale=.75]{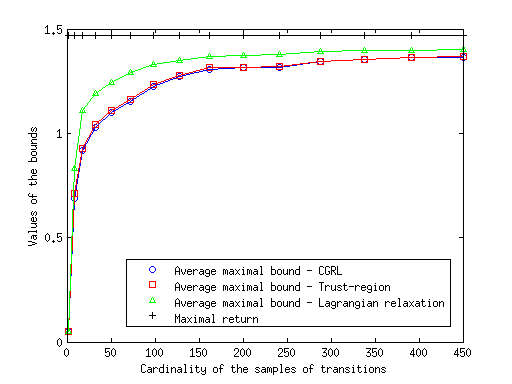}
\caption{Average values $A_{CGRL}(c_i)$, $A_{TR}(c_i)$ and $A_{LD}(c_i)$ of the bounds computed from all samples of transitions $\mathcal F_{c_i,k} \quad k \in \{1, \ldots, 100\}$ of cardinality $c_i = 2 i^2 .$ \label{figure:linear_cardinality_uniform_bounds}}
\end{center}
\end{figure}
\begin{figure}[h!]
\begin{center}
\includegraphics[scale=.75]{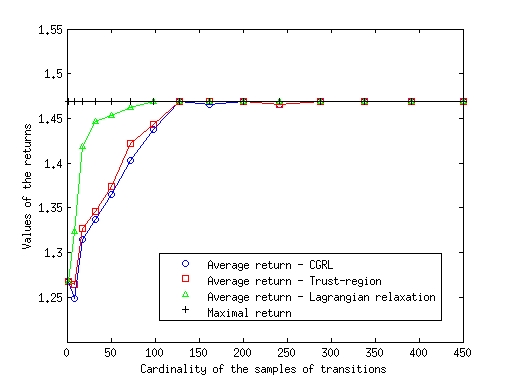}
\caption{Average values $J_{CGRL}$, $J_{TR}$ and $J_{LD}$ of the return of the  bound-optimal sequences of actions computed from all samples of transitions $\mathcal F_{c_i,k} \quad k \in \{1, \ldots, 100\}$ of cardinality $c_i = 2 i^2 .$ \label{figure:linear_cardinality_uniform_returns}}
\end{center}
\end{figure}

In order to observe the influence of the dispersion of the state-action points of the transitions on the quality of the bounds, we propose the following protocol. For each cardinality $c_i =2 i^2 , i=1 , \ldots , 15$, we generate 100 samples of transitions $\mathcal F_{c_i,1} , \ldots, \mathcal F_{c_i,100} $ using a uniform probability distribution over the space $[0 , 1]^{d} \times \mathcal U $. For each sample of transition $\mathcal F_{c_i,k} \quad i \in \{1, \ldots , 15 \} , k \in \{1 , \ldots, 100 \}$, we compute the maximal CGRL bound  $ B^{(*)}_{CGRL}\left(\mathcal F_{c_i,k}\right)$, the maximal Trust-region bound $ B^{(*)}_{TR}\left(\mathcal F_{c_i,k}\right)$ and the maximal Lagrangian relaxation bound $B^{(*)}_{LD}\left(\mathcal F_{c_i,k}\right)$. We then compute the average values of the maximal CGRL, Trust-region and Lagrangian relaxation bounds :
\begin{eqnarray*}
\forall i \in \{ 1, \ldots, 15 \}, \qquad A_{CGRL}(c_i) &=&  \frac{1}{100}\sum_{k=1}^{100} B^{(*)}_{CGRL}\left(\mathcal F_{c_i,k}\right)   \\
A_{TR}(c_i) &=& \frac{1}{100}\sum_{k=1}^{100} B^{(*)}_{TR}\left(\mathcal F_{c_i,k}\right)   \\
A_{LD}(c_i) &=& \frac{1}{100}\sum_{k=1}^{100} B^{(*)}_{LD}\left(\mathcal F_{c_i,k}\right)   
\end{eqnarray*}
and we report in Figure \ref{figure:linear_cardinality_uniform_bounds}  the values $A_{CGRL}(c_i)$ (resp. $A_{TR}(c_i)$ and $A_{LD}(c_i)$)  as a function of the cardinality $c_i$ of the samples of transitions. We also report in Figure \ref{figure:linear_cardinality_uniform_returns}  the average returns of the bound-optimal sequences of actions  $\left(u_0 , u_1  \right)^{CGRL}_{\mathcal F_{c_i,k} }$, $\left(u_0 , u_1  \right)^{TR}_{\mathcal F_{c_i,k} }$ and $\left(u_0 , u_1  \right)^{LD}_{\mathcal F_{c_i,k}}$:
\begin{eqnarray*}
\forall i \in \{ 1, \ldots, 15 \}, \qquad J_{CGRL}(c_i) &=& \frac{1}{100}\sum_{k=1}^{100} J^{\left(u_0 , u_1  \right)^{CGRL}_{\mathcal F_{c_i,k}}}_2 \\
J_{TR}(c_i) &=& \frac{1}{100}\sum_{k=1}^{100} J^{\left(u_0 , u_1  \right)^{TR}_{\mathcal F_{c_i,k}}}_2 \\
J_{LD}(c_i) &=& \frac{1}{100}\sum_{k=1}^{100} J^{\left(u_0 , u_1  \right)^{LD}_{\mathcal F_{c_i,k}}}_2  \ .
\end{eqnarray*}
as a function of the cardinality $c_i$ of the samples of transitions.

We observe that, on average, the Lagrangian relaxation bound is much tighter that the Trust-region and the CGRL bounds. The CGRL bound and the Trust-region bound remain very close on average, which illustrates, in a sense, Corollary \ref{corollary:TR_et_CGRL_super_close}. Moreover, we also observe that the bound-optimal sequences of actions $\left(u_0 , u_1  \right)^{LD}_{\mathcal F_{c_i,k}}$ better perform on average.

\section{Conclusions}
\label{section:conclusions}

We have considered in this paper the problem of computing $\min \max$ policies for deterministic, Lipschitz continuous batch mode reinforcement learning. First, we have shown that this $\min \max$ problem is NP-hard. Afterwards, we have proposed for the two-stage case two relaxation schemes. Both have been extensively studied and, in particular, they have been shown to perform better than the CGRL algorithm that has been introduced earlier to address this min-max generalization problem.

A natural extension of this work would be to investigate how the proposed relaxation schemes could be extended to the $T$-stage ($T \geq 3$) framework. Lipschitz continuity assumptions are common in a batch mode reinforcement learning setting, but one could imagine developing $\min \max$ strategies in other types of environments that are not necessarily Lipschitzian, or even not continuous. Additionnaly, it would also be interesting to extend the resolution schemes proposed in this paper to problems with very large/continuous action spaces.

\subsection*{Acknowledgements}
Raphael Fonteneau is a Post-doctoral fellow of the FRS-FNRS (Funds for Scientific Research). This paper  presents research   results  of  the   Belgian  Network  DYSCO (Dynamical Systems, Control and Optimization) funded  by the Interuniversity  Attraction Poles Programme, initiated  by the Belgian State, Science  Policy  Office. The authors thank Yurii Nesterov for pointing out the idea of using Lagrangian relaxation. The scientific responsibility   rests with its authors.

\appendix

\section{Proof of Lemma \ref{lemma:n0_n1_1}}
\label{appendix:proof_n0_n1_1}

\begin{proof}
For conciseness, we denote $\left(x^{(u_0),k_0},r^{(u_0),k_0},y^{(u_0),k_0} \right)$ ( resp.\\ $\left(x^{(u_1),k_1},r^{(u_1),k_1},y^{(u_1),k_1} \right)$) by $\left(x^{0},r^{0},y^{0}\right)$ (resp. $\left(x^{1},r^{1},y^{1} \right)$), and $\lambda_1$ (resp. $\mu_1$) by $\lambda$ (resp.  $\mu$). We assume that $x_0 \neq  x^{0}$ and $x^{1} \neq y^{0}$ otherwise the problem is trivial.
\begin{itemize}
\item Trust-region solution.
\end{itemize}
According to Definition \ref{definition_TR}, we have:
\begin{eqnarray*}
B''^{(u_0,u_1),k_0,k_1}_{TR}(\mathcal F) &=&   r^{1}-L_{\rho} \left\| \mathbf{\hat x_1^*}(k_0,k_1) -x^{1} \right\|,
\end{eqnarray*}
where
\begin{eqnarray*}
\mathbf{\hat x_1^*}(k_0,k_1) = y^{0} +  L_f\frac{ \left\| x_0 - x^{0} \right\|}{\left\|y^{0}-x^{1} \right\|} \left(y^{0}-x^{1}\right)  , \nonumber
\end{eqnarray*}
which writes
\begin{eqnarray*}
B''^{(u_0,u_1),k_0,k_1}_{TR}(\mathcal F) &=& r^{1}-L_{\rho} \left\| y^{0} +  L_f\frac{ \left\| x_0 - x^{0} \right\|}{\left\|y^{0}-x^{1} \right\|} \left(y^{0}-x^{1}\right)  -x^{1} \right\| ,  \\
&=& r^{1}-L_{\rho} \left\|y^{0}-x^{1} \right\| \left(  1 +   L_f\frac{ \left\| x_0 - x^{0} \right\|}{\left\|y^{0}-x^{1} \right\|}  \right) \\
&=& r^{1} - L_{\rho} \left\|y^{0}-x^{1} \right\|  -  L_\rho L_f \left\| x_0 - x^{0} \right\|  
\end{eqnarray*}

\begin{itemize}
\item Lagrangian relaxation based solution.
\end{itemize}

According to Equation (\ref{socp}), we can write:
\begin{eqnarray*}
B''^{(u_0,u_1),k_0,k_1}_{LD}(\mathcal F) &=&  \underset{\lambda\in \R_+, \mu \in \R_+} {\max} \;
\frac{-\left\|L_{\rho}^2 x^{1} \mu - y^{0} \lambda \right\|^2}{- \mu L_{\rho}^2+ \lambda} - \frac{\left(1-2  r^1 \mu\right)^2 }{4 \mu}  \\
&&+  \lambda \left(\left\|y^{0}\right\|^2 -L_f^2 \left\|x^{0} - x_0 \right\|^2 \right) 
 \quad + \mu \left(\left(r^{1}\right)^2- L_{\rho}^2 \left\|x^{1}\right\|^2 \right)
\end{eqnarray*}
subject to 
\begin{eqnarray*}
\mu  &>& 0 \\
\lambda &>& \mu L_{\rho}^2
\end{eqnarray*}
We denote by $\mathcal L(\lambda, \mu)$ the quantity:
\begin{eqnarray*}
\mathcal L(\lambda, \mu) &=& \frac{-\left\|L_{\rho}^2 x^{1} \mu - y^{0} \lambda \right\|^2}{- \mu L_{\rho}^2+ \lambda}  - \frac{\left(1-2  r^1 \mu\right)^2 }{4 \mu}  +  \lambda \left(\left\|y^{0}\right\|^2 -L_f^2 \left\|x^{0} - x_0 \right\|^2 \right) \\
&&+ \mu \left(\left(r^{1}\right)^2- L_{\rho}^2 \left\|x^{1}\right\|^2 \right)
\end{eqnarray*}
Let $\lambda$ and $\mu$ be such that $\lambda > \mu L_{\rho}^2$. Since the Trust-region solution to $( \mathcal P''^{(u_0,u_1)}_{TR}(k_0,k_1) )$ is optimal, and by property of the Lagrangian relaxation \cite{LeMarechal1996}, one has:
\begin{eqnarray}
\mathcal L(\lambda, \mu) \leq B''^{(u_0,u_1),k_0,k_1}_{TR}(\mathcal F) \ . \label{lagrangian_inequality}
\end{eqnarray}
In order to prove the lemma, it is therefore sufficient to determine two values $\lambda_0$ and $\mu_0$ such that the inequality (\ref{lagrangian_inequality}) is an equality. By differentiating $\mathcal L(\lambda,\mu)$, we obtain, after a long calculation  (that we omit here):
\begin{eqnarray*}
\left( \left\{ \begin{matrix} \frac {\partial \mathcal L(\lambda, \mu)} {\partial \lambda} = 0  \\  \frac { \partial \mathcal L(\lambda, \mu)} {\partial \mu} = 0  \\  \lambda > \mu L_{\rho}^2 \end{matrix}  \right. \right) \implies  \left\{  \begin{matrix}  \lambda = \frac {L_\rho} {2 L_f \left\| x^0 -  x_0 \right\|} \ ,\\ \mu = \frac {1} { 2 L_\rho \left(  \left\|  y^0 - x^1 \right\| + L_f \left\|  x^0 - x_0 \right\|  \right)} \ .  \end{matrix}  \right.
\end{eqnarray*}
We denote by $\lambda_0$ and $\mu_0$ the following values for the dual variables:
\begin{eqnarray*}
\lambda_0 &\triangleq& \frac {L_\rho} {2 L_f \left\| x^0 - x_0 \right\|} \ ,  \\
\mu_0 &\triangleq& \frac {1} { 2 L_\rho \left(  \left\|  y^0 - x^1 \right\| + L_f \left\|  x^0 - x_0 \right\|  \right)} \ .
\end{eqnarray*}
We have:
\begin{eqnarray*}
\mu_0 = \frac {1} { 2 L_\rho \left(  \left\|  y^0 - x^1 \right\| + L_f \left\|  x^0 - x_0 \right\|  \right)}  &>& 0  \\
\frac {\lambda_0} {\mu_0} = L_\rho^2 \left( 1 +  \frac { \left\| y^0 - x^1  \right\|} { L_f \left\| x^0 - x_0 \right\|} \right) &>& L_\rho^2 \ .
\end{eqnarray*}
In the following, we denote $\left( 1 +  \frac {\left\| y^0 - x^1  \right\|} { L_f \left\| x^0 - x_0 \right\|} \right)$ by $K$. We now give the expression of $\mathcal L(\lambda_0,\mu_0)$ using only $\mu_0$ and $K$:
\begin{eqnarray*}
\mathcal L(\lambda_0,\mu_0) &=& -\frac{L_\rho^4 \mu^2_0 \left\| x^{1} -  K y^{0}  \right\|^{2}}{\mu_0 L_\rho^2 ( - 1 + K )} - \frac{1}{4 \mu_0} + r^{1}  - \left( r^{1} \right)^{2} \mu_0 \\
&&  + \mu_0 K L_\rho^2 \left( \left\|  y^0 \right\|^2 - L^2_f \left\| x^{0} - x_0  \right\|^2  \right) + \mu_0 \left(  \left( r^{1} \right)^2 - L_\rho^2 \left\|  x^{1}  \right\|^2  \right) \\
&=& - \frac { L_\rho^2 \mu_0  \left\|  x^{1} -  K y^0  \right\|^{2} } { K - 1 }     - \frac {1} {4 \mu_0} + r^{1} \\
&&+ L_\rho^2 \mu_0 K \left(  \left\| y^0 \right\|^2 -  L_f^2 \left\|  x^0 -  x_0  \right\|^2  \right) - L_\rho^2 \mu_0   \left\| x^1 \right\|^2  \ .
\end{eqnarray*}
Using the fact that $ x^{1} -  K y^0  =  x^{1} - y^{0} -  (K-1) y^0$, we can write: 
\begin{eqnarray*}
\mathcal L(\lambda_0,\mu_0) &=&  - \frac { L_\rho^2 \mu_0 }{K-1} \left( \left\|  x^{1} - y^{0} \right\|^{2} + (K-1)^2 \left\|y^0  \right\|^{2} - 2 (K-1)(x^{1} - y^{0})^{T}y^{0}   \right)  \\
&&- \frac{1}{4 \mu_0} + r^{1} + L_{\rho}^{2} \mu_0 K \left(  \left\| y^{0}  \right\|^{2} - L_f^2 \left\| x^{0} - x_0 \right\|^{2}  \right) - L_\rho^2 \left\| x^{1} \right\|^{2}
\end{eqnarray*}
and
\begin{eqnarray*}
\mathcal L(\lambda_0,\mu_0) &=&  - \frac { L_\rho^2 \mu_0 }{K-1}  \left\|  x^{1} - y^{0} \right\|^{2} - L_\rho^2 \mu_0(K-1) \left\| y^{0} \right\|^{2} + 2 L_\rho^2 \mu_0 (x^{1} - y^{0})^{T} y^{0} \\
&&- \frac{1}{4 \mu_0} + r^{1} + L_{\rho}^{2} \mu_0 K \left(  \left\| y^{0}  \right\|^{2} - L_f^2 \left\| x^{0} -  x_0 \right\|^{2}  \right) - L_\rho^2 \left\| x^{1} \right\|^{2} \ .
\end{eqnarray*}
From there,
\begin{eqnarray*}
\mathcal L(\lambda_0,\mu_0) &=&  - \frac { L_\rho^2 \mu_0 }{K-1}  \left\|  x^{1} - y^{0} \right\|^{2} + \left\| y^0 \right\|^{2} \left(  - L_\rho^2 \mu_0 (K-1) - 2 L_\rho^2 \mu_0 + L_\rho^2 \mu_0 K \right)  \\
&&- 2 L_\rho^2 \mu_0 (x^{1})^{T} y^{0}  - \frac{1}{4 \mu_0} + r^{1} + L_{\rho}^{2} \mu_0 K \left( - L_f^2 \left\| x^{0} - x_0 \right\|^{2}  \right) - L_\rho^2 \left\| x^{1} \right\|^{2}
\end{eqnarray*}
and, since $\left(  - L_\rho^2 \mu_0 (K-1) - 2 L_\rho^2 \mu_0 + L_\rho^2 \mu_0 K \right) = -  L_\rho^2 \mu_0 $, we have that
\begin{eqnarray*}
\mathcal L(\lambda_0,\mu_0) &=&  - \frac { L_\rho^2 \mu_0 }{K-1}  \left\|  x^{1} - y^{0} \right\|^{2} - L_\rho^2 \mu_0  \left\| y^0 \right\|^{2} - 2 L_\rho^2 \mu_0 (x^{1})^{T} y^{0} \\
&&- \frac{1}{4 \mu_0} + r^{1} + L_{\rho}^{2} \mu_0 K \left( - L_f^2 \left\| x^{0} -  x_0 \right\|^{2}  \right) - L_\rho^2 \left\| x^{1} \right\|^{2}
\end{eqnarray*}
and
\begin{eqnarray*}
\mathcal L(\lambda_0,\mu_0) &=&  - \frac { L_\rho^2 \mu_0 }{K-1}  \left\|  x^{1} - y^{0} \right\|^{2} - L_\rho^2 \mu_0 \left( \left\| y^{0}  \right\|^{2} +  \left\| x^{1} \right\|^{2} -2(x^{1})^{T} y^{0} \right) \\
&&- \frac{1}{4 \mu_0} + r^{1} + L_{\rho}^{2} \mu_0 K \left( - L_f^2 \left\| x^{0} -  x_0 \right\|^{2}  \right)  \ .
\end{eqnarray*}
Since $\left( \left\| y^{0}  \right\|^{2} +  \left\| x^{1} \right\|^{2} -2(x^{1})^{T} y^{0} \right) = \left\|  x^{1} - y^{0} \right\|^{2}$, we have:
\begin{eqnarray*}
\mathcal L(\lambda_0,\mu_0) &=&  - \frac { L_\rho^2 \mu_0 }{K-1}  \left\|  x^{1} - y^{0} \right\|^{2} - L_\rho^2 \mu_0 \left\| x^{1} - y^{0} \right\|^{2} \\
&&- \frac{1}{4 \mu_0} + r^{1} + L_{\rho}^{2} \mu_0 K \left( - L_f^2 \left\| x^{0} - x_0 \right\|^{2}  \right)  \\
&=& - L_\rho^2 \mu_0 \left\| x^{1} - y^{0}  \right\|^{2} \frac {K} {K-1} \\
&&- \frac{1}{4 \mu_0} + r^{1} + L_{\rho}^{2} \mu_0 K \left( - L_f^2 \left\| x^{0} -  x_0 \right\|^{2}  \right) \\
&=& - K L_\rho^2 \mu_0 \left(  \frac{ \left\| x^{1} - y^{0}  \right\|^{2}}{K-1} + L_f^2 \left\| x^{0} -  x_0  \right\|^2  \right) -\frac{1}{4 \mu_0} + r^{1} \ .
\end{eqnarray*}
Since $K \triangleq \left( 1 +  \frac {\left\| y^0 - x^1  \right\|} { L_f \left\| x^0 -  x_0 \right\|} \right) $, we have:
\begin{eqnarray*}
\mathcal L(\lambda_0,\mu_0) &=& - \left( 1 +  \frac {\left\| y^0 - x^1  \right\|} { L_f \left\| x^0 - \hat x_0 \right\|} \right) L_\rho^2 \mu_0 \left(  \frac{\left\| x^{1} - y^{0}  \right\|^{2}}{\left( 1 +  \frac {\left\| y^0 - x^1  \right\|} { L_f \left\| x^0 -  x_0 \right\|} \right)-1} + L_f^2 \left\| x^{0} -  x_0  \right\|^{2}  \right) \\
&& -\frac{1}{4 \mu_0} + r^{1} \\
= &-& \frac{L_\rho^2 \mu_0 \left( L_f \left\| x^{0} -  x_0 \right\| + \left\| y^{0} - x^{1} \right\|  \right)}{L_f \left\| x^{0} - x_0 \right\|} \left(  \frac{L_f \left\| x^{0} -  x_0  \right\| \left\| x^{1} - y^{0}  \right\|^{2}}{\left\| x^{1} - y^{0}  \right\|} + L_f^{2} \left\| x^{0} -  x_0  \right\|^2 \right) \\
&&- \frac{1}{4 \mu_0} + r^1\\
= &-& L_\rho^2 \mu_0 \left( L_f \left\| x^{0} -  x_0 \right\| + \left\| y^{0} - x^{1} \right\|  \right) \left(  \left\| x^{1} - y^{0}  \right\| + L_f \left\|  x^{0} - x_0 \right\| \right) - \frac{1}{4 \mu_0} + r^1 \ .
\end{eqnarray*}
Since $\mu_0 =\frac {1} { 2 L_\rho \left(  \left\|  y^0 - x^1 \right\| + L_f \left\|  x^0 - x_0 \right\|  \right)}$, we finally obtain:
\begin{eqnarray*}
\mathcal L(\lambda_0,\mu_0) &=& - \frac {L\rho} {2}  \left(  \left\|  y^0 - x^1 \right\| + L_f \left\|  x^0 -  x_0 \right\|  \right)   - \frac {L\rho} {2}  \left(  \left\|  y^0 - x^1 \right\| + L_f \left\|  x^0 -  x_0 \right\|  \right) + r^{1} \\
&=& r^{1}  - L_\rho \left\| y^{0} - x^{1}  \right\| - L_\rho L_f \left\| x^{0} - x_0  \right\| \\
&=& B''^{(u_0,u_1),k_0,k_1}_{TR}(\mathcal F) \ ,
\end{eqnarray*}
which ends the proof.
\end{proof}

\bibliographystyle{siam.bst}
\bibliography{all}

\end{document}